\documentclass[12pt]{article}
\usepackage{enumerate}
\usepackage{amsfonts}
\usepackage{amsmath}
\usepackage{amssymb}
\usepackage{amsthm}

\def\H{\mathcal{H}}

\def\S{\mathfrak{S}}

\def\B{\mathfrak{B}}

\newcounter{defin}  \newcounter{lemma}  \newcounter{theorem}
\newcounter{property} \newcounter{corol}  \newcounter{remark} \newcounter{example}

\newenvironment{lemma}{\par\refstepcounter{lemma}
     \textbf{Lemma \thelemma.} }{\rm\par}
\newenvironment{theorem}{\par\refstepcounter{theorem}
     \textbf{Theorem \thetheorem.}\ }{\rm\par}

\newenvironment{corollary}{\par\refstepcounter{corol}
     \textbf{Corollary \thecorol.} }{\rm\par}

\newenvironment{remark}{\par\refstepcounter{remark}
     \textbf{Remark \theremark.}}{\rm\par}

\begin{document}

\title{On n-partite superactivation of quantum channel capacities}
\author{M.E. Shirokov\footnote{Steklov Mathematical Institute, RAS, Moscow, email:msh@mi.ras.ru}}
\date{}
\maketitle \vspace{-10pt}

\begin{abstract}
A generalization of the superactivation of  quantum channel
capacities to the case of $\,n>2\,$ channels is considered. An
explicit example of such superactivation for the 1-shot quantum
zero-error capacity is constructed for $n=3$.

Some implications of this example and its reformulation on terms of
quantum measurements are described.
\end{abstract}
\maketitle

\section{General observations}

The superactivation of quantum channel capacities is one of the most
impressive quantum effects having no classical counterpart. It means
that the particular capacity $C$ of the tensor product of two
quantum channels $\Phi_1$ and $\Phi_2$ may be positive despite the
same capacity of each of these channels is zero, i.e.
\begin{equation}\label{s-2}
C(\Phi_1\otimes\Phi_2)>0\quad \textrm{while} \quad
C(\Phi_{1})=C(\Phi_{2})=0.
\end{equation}

This effect was originally observed by G.Smith and J.Yard for the
case of quantum $\varepsilon$-error capacity \cite{S&Y}. Then the
possibility of superactivation of other capacities, in particular,
zero-error capacities was shown \cite{CCH,C&S,Duan,SSY}.

It seems reasonable to consider the generalization of the above
effect to the case of $\,n>2\,$ channels $\Phi_1,...,\Phi_n$
consisting in the following property
\begin{equation}\label{s-n}
C(\Phi_1\otimes\ldots\otimes\Phi_n)>0\quad \textrm{while} \quad
C(\Phi_{i_1}\otimes\ldots\otimes\Phi_{i_k})=0
\end{equation}
for any proper subset $\Phi_{i_1},...,\Phi_{i_k}$ ($k<n$) of the set
$\Phi_1,...,\Phi_n$. This property can be called \emph{$n$-partite
superactivation} of the capacity $C$.

Property (\ref{s-n}) means that all the channels $\Phi_1,...,\Phi_n$
are required to transmit (classical or quantum) information by using
the protocol corresponding to the capacity $C$, i.e. excluding any
channel from the set $\Phi_1,...,\Phi_n$ makes other channels
useless for information transmission.\smallskip

The obvious difficulty in finding channels $\Phi_1,...,\Phi_n$
demonstrating property (\ref{s-n}) for given capacity $C$ consists
in necessity to prove the vanishing of
$\,C(\Phi_{i_1}\otimes\ldots\otimes\Phi_{i_k})\,$ for any subset
$\Phi_{i_1},...,\Phi_{i_k}$.\smallskip

If $C$ is the 1-shot capacity of some protocol of information
transmission and $\Phi_i=\Phi$ for all $i=\overline{1,n}$ then
(\ref{s-n}) means that the $n$-shot capacity of this protocol is
positive while the corresponding  $(n-1)$-shot capacity is
zero.\smallskip

In \cite{NBC} it is shown how to construct for any $n$ a channel
$\Psi_n$ such that
\begin{equation}\label{nbc}
\bar{Q}_0(\Psi^{\otimes n}_n)=0\quad \textrm{and} \quad
\bar{Q}_0(\Psi^{\otimes m}_n)>0,
\end{equation}
where $\bar{Q}_0$ is the 1-shot quantum zero-error capacity and
$\,m\,$ is a natural number satisfying the inequality
$n/m\leq2\ln(3/2)/\pi$ (implying $m>n$). It follows that there is
$\,\tilde{n}>n\,$ not exceeding $m$ such that (\ref{s-n}) holds for
$n=\tilde{n}$, $C=\bar{Q}_0$ and
$\Phi_1=...=\Phi_{\tilde{n}}=\Psi_n$. Unfortunately, we can not
specify the number $\,\tilde{n}\,$ in that construction.

In this paper we modify the example in \cite{NBC} (by extending its
noncommutative graph) to construct a family of channels
$\{\Phi_{\theta}\}$ with $d_A=4$  and $d_E=3$ having the following
property
\begin{equation}\label{3-p-s}
\bar{Q}_0(\Phi_{\theta_1}\otimes
\Phi_{\theta_2}\otimes\Phi_{\theta_3})>0\quad \textrm{while}
\quad\bar{Q}_0(\Phi_{\theta_i}\otimes \Phi_{\theta_j})=0\quad
\forall\, i\neq j,
\end{equation}
where $\theta_1,\theta_2,\theta_3$ are positive numbers such that
$\,\theta_1+\theta_2+\theta_3=\pi\,$. Thus, the channels
$\Phi_{\theta_1},\Phi_{\theta_2},\Phi_{\theta_3}$ demonstrate the
3-partite superactivation of the 1-shot quantum zero-error capacity.

Property (\ref{3-p-s}) means that all the  channels
$\Phi_{\theta_i}$ and all the bipartite channels
$\Phi_{\theta_i}\otimes \Phi_{\theta_j}$ have no ideal (noiseless or
reversible) subchannels, but the tripartite channel
$\Phi_{\theta_1}\otimes \Phi_{\theta_2}\otimes\Phi_{\theta_3}$ has.

By using the observation in \cite[Section 4]{Sh&Sh} superactivation
property (\ref{3-p-s}) can be reformulated in terms of quantum
measurements theory as the existence of quantum observables
$\mathcal{M}_{\theta_1},\mathcal{M}_{\theta_2},\mathcal{M}_{\theta_3}$
such that all the observables $\mathcal{M}_{\theta_i}$ and all the
bipartite observables $\mathcal{M}_{\theta_i}\otimes
\mathcal{M}_{\theta_j}$ have no indistinguishable subspaces but the
tripartite observable $\mathcal{M}_{\theta_1}\otimes
\mathcal{M}_{\theta_2}\otimes \mathcal{M}_{\theta_3}$ has.

\section{Preliminaries}

Let
$\Phi:\mathfrak{S}(\mathcal{H}_A)\rightarrow\mathfrak{S}(\mathcal{H}_B)$
be a quantum channel, i.e. a  completely positive trace-preserving
linear map \cite{H-SCI,N&Ch}. It has the Kraus representation
\begin{equation}\label{Kraus-rep}
\Phi(\rho)=\sum_{k}V_{k}\rho V^{*}_{k},
\end{equation}
where $V_{k}$ are linear operators from $\mathcal{H}_A$ into
$\mathcal{H}_B$ such that $\sum_{k}V^{*}_{k}V_{k}=I_{\H_A}$. The
minimal number of summands in such representation is called
\emph{Choi rank} of $\Phi$ and is denoted $d_E$ (while
$d_A\doteq\dim\H_A$ and $d_B\doteq\dim\H_B$).

The 1-shot quantum zero-error capacity $\bar{Q}_0(\Phi)$ of a
channel $\Phi$ is defined as $\;\sup_{\H\in
q_0(\Phi)}\log_2\dim\H\,$, where $q_0(\Phi)$ is the set of all
subspaces $\H_0$ of $\H_A$ on which the channel $\Phi$ is perfectly
reversible (in the sense that there is a channel $\Theta$ such that
$\Theta(\Phi(\rho))=\rho$ for all states $\rho$ supported by
$\H_0$). Any subspace $\H_0\in q_0(\Phi)$ is called \emph{error
correcting code} for the channel $\Phi$ \cite{W&Co, H-SCI}.

The (asymptotic) quantum zero-error  capacity is defined by
regularization: $Q_0(\Phi)=\sup_n n^{-1}\bar{Q}_0(\Phi^{\otimes n})$
\cite{ZEC, C&S, W&Co}.

The quantum zero-error capacity of a channel $\Phi$ is determined by
its \emph{noncommutative graph} $\mathcal{G}(\Phi)$, which can be
defined as the subspace of $\mathfrak{B}(\mathcal{H}_A)$ spanned by
the operators $V^*_kV_l$, where $V_k$ are operators from any Kraus
representation (\ref{Kraus-rep}) of $\Phi$ \cite{W&Co}. In
particular, the Knill-Laflamme error-correcting condition \cite{K-L}
implies the following lemma.\smallskip

\begin{lemma}\label{trans-l+}
\emph{A set $\{\varphi_k\}_{k=1}^d$ of unit orthogonal vectors in
$\H_A$ is a basis of error-correcting code for a channel
$\,\Phi:\S(\H_A)\rightarrow\S(\H_B)$ if and only if
\begin{equation}\label{operators+}
\langle \varphi_l|A|\varphi_k\rangle=0\quad\textit{and}\quad \langle
\varphi_l|A|\varphi_l\rangle=\langle
\varphi_k|A|\varphi_k\rangle\quad\forall
A\in\mathfrak{L},\;\,\forall k\neq l,
\end{equation}
where $\mathfrak{L}$ is any subset of $\,\B(\H_A)$ such that
$\,\mathrm{lin}\mathfrak{L}=\mathcal{G}(\Phi)$.}
\end{lemma}\smallskip

Since a subspace $\mathfrak{L}$ of the algebra $\mathfrak{M}_n$ of
$n\times n$ matrices is a noncommutative graph of a particular
channel if and only if
\begin{equation}\label{L-cond}
\mathfrak{L}\;\,\textup{is
symmetric}\;\,(\mathfrak{L}=\mathfrak{L}^*)\;\,\textup{and contains
the unit matrix}
\end{equation}
(see Lemma 2 in \cite{Duan} or the Appendix in \cite{Sh&Sh}), Lemma
\ref{trans-l+} shows that one can "construct" a channel $\Phi$ with
$\dim\H_A=n$ having positive (correspondingly, zero) 1-shot quantum
zero-error  capacity by taking a subspace
$\mathfrak{L}\subset\mathfrak{M}_n$ satisfying (\ref{L-cond}) for
which the following condition is valid (correspondingly, not valid)
\begin{equation}\label{l-3-c}
\exists\varphi,\psi\in[\mathbb{C}^n]_1 \;\;\textup{s.t.}\;\; \langle
\psi|A|\varphi\rangle=0\;\;\textup{and}\;\; \langle
\varphi|A|\varphi\rangle=\langle \psi|A|\psi\rangle\quad\forall
A\in\mathfrak{L},
\end{equation}
where $[\mathbb{C}^n]_1$ is the unit sphere of $\mathbb{C}^n$.

\section{Example of 3-partite superactivation}

For given $\theta\in(-\pi,\pi]$ consider the 8-D subspace
\begin{equation}\label{L-def++}
\mathfrak{N}_{\theta} = \left\{M=\left[\begin{array}{cccc}
a &  b & e &  f\\
c &  d & f &  \bar{\gamma} e\\
g &  h & a &  b\\
h &  \gamma g & c &  d
\end{array}\right]\!,\;\; a,b,c,d,e,f,g,h\in\mathbb{C}\;\right\}
\end{equation}
of $\mathfrak{M}_4$ satisfying condition (\ref{L-cond}), where
$\gamma=\exp\left[\,\mathrm{i}\theta\,\right]$. This subspace is an
extension of the 4-D subspace $\mathfrak{L}_{\theta}$ used in
\cite{NBC}, i.e. $\mathfrak{L}_{\theta}\subset\mathfrak{N}_{\theta}$
for each $\theta$.\footnote{In contrast to this paper
$\,\gamma=\exp\left[\,\mathrm{i}\,\theta/2\,\right]\,$ is used in
\cite{NBC}.}

Denote by $\widehat{\mathfrak{N}}_{\theta}$ the set of all channels
whose noncommutative graph coincides with $\mathfrak{N}_{\theta}$.
In \cite[the Appendix]{Sh&Sh} it is shown how to explicitly
construct pseudo-diagonal channels in
$\widehat{\mathfrak{N}}_{\theta}$ with $d_A=4$ and $d_E\geq3$ (since
$\dim \mathfrak{N}_{\theta}=8\leq3^2$).
\medskip

\begin{theorem}\label{sqc}  \emph{Let  $\,\Phi_{\theta}$ be a  channel in $\,\widehat{\mathfrak{N}}_{\theta}$ and $n\in\mathbb{N}$ be arbitrary.}\medskip

A) \emph{$\bar{Q}_0(\Phi_{\theta})>0\;$  if and only if
$\;\theta=\pi\,$ and $\;\bar{Q}_0(\Phi_{\pi})=1\;$.}

\medskip

B) \emph{If $\;\theta_1+\ldots+\theta_n=\pi(\mathrm{mod}\,2\pi)\,$
then
$\;\bar{Q}_0(\Phi_{\theta_1}\otimes\ldots\otimes\Phi_{\theta_n})>0\;$
and $\,2\textrm{-}D$ error-correcting code for the channel
$\,\Phi_{\theta_1}\otimes\ldots\otimes\Phi_{\theta_n}$ is spanned by
the vectors
\begin{equation}\label{main-vec}
|\varphi\rangle=\textstyle{\frac{1}{\sqrt{2}}}\left[\;|1\ldots
1\rangle+\mathrm{i}\,|2\ldots 2\rangle\,\right],\;
|\psi\rangle=\textstyle{\frac{1}{\sqrt{2}}}\left[\;|3\ldots
3\rangle+\mathrm{i}\,|4\ldots 4\rangle\,\right]\!,\vspace{5pt}
\end{equation}
where $\{|1\rangle, \ldots, |4\rangle\}$ is the canonical basis in
$\mathbb{C}^4$.} \medskip

$\mathbf{C_2}$) \emph{If $\;|\theta_1|+|\theta_2|<\pi\,$ then
$\;\bar{Q}_0(\Phi_{\theta_1}\otimes\Phi_{\theta_2})=0\;$.}

\medskip

$\mathrm{C_n}$) \emph{If $\;|\theta_1|+\ldots+|\theta_n|\leq
2\ln(3/2)\,$ then
$\;\bar{Q}_0(\Phi_{\theta_1}\otimes\ldots\otimes\Phi_{\theta_n})=0\;$.}
\end{theorem}

\medskip

Assertion $\mathrm{C_2}$ is the main progress of this theorem in
comparison with Theorem 1 in \cite{NBC}. It complements assertion B
with $n=2$. It is the proof of assertion $\mathrm{C_2}$ that
motivates the extension
$\mathfrak{L}_{\theta}\rightarrow\mathfrak{N}_{\theta}$.\smallskip

\begin{remark}\label{sqc-r} Since assertion $\mathrm{C_n}$ is proved by using quite coarse
estimates, the other assertions of Theorem 1 make it reasonable to
conjecture validity of the following strengthened version:\smallskip

$\mathrm{C^{*}_n}$) \emph{If $\;|\theta_1|+\ldots+|\theta_n|< \pi\,$
then
$\;\bar{Q}_0(\Phi_{\theta_1}\otimes\ldots\otimes\Phi_{\theta_n})=0.\;$}
\smallskip

\noindent The below proof of  $\mathrm{C_2}$ shows difficulty of the
direct proof of this conjecture.
\end{remark}

\medskip

Theorem \ref{sqc} implies the following  example of $3$-partite
superactivation of 1-shot quantum zero-error capacity.
\smallskip

\begin{corollary}\label{sqc-c} \emph{Let $\,\theta_1,\theta_2,\theta_3$ be positive numbers
such that $\,\theta_1+\theta_2+\theta_3=\pi$. Then
$$
\bar{Q}_0(\Phi_{\theta_1}\otimes
\Phi_{\theta_2}\otimes\Phi_{\theta_3})>0\quad \textit{while}
\quad\bar{Q}_0(\Phi_{\theta_i}\otimes \Phi_{\theta_j})=0\quad
\forall\, i\neq j.
$$
$\,2\textrm{-}D$ error-correcting code for the channel
$\,\Phi_{\theta_1}\otimes \Phi_{\theta_2}\otimes\Phi_{\theta_3}$ is
spanned by the vectors}
\begin{equation}\label{main-vec+}
|\varphi\rangle=\textstyle{\frac{1}{\sqrt{2}}}\left[\;|111\rangle+\mathrm{i}\,|222\rangle\,\right],\;
|\psi\rangle=\textstyle{\frac{1}{\sqrt{2}}}\left[\;|333\rangle+\mathrm{i}\,|444\rangle\,\right]\!.\vspace{5pt}
\end{equation}
\end{corollary}

If conjecture $\mathrm{C^{*}_n}$ in Remark \ref{sqc} is valid for
some $n>2$ then the similar assertion holds for $\,n+1\,$ channels
$\Phi_{\theta_1},\ldots,\Phi_{\theta_{n+1}}$. This would give an
example of $(n+1)$-partite superactivation of 1-shot  quantum
zero-error capacity.\medskip

For each $\theta$ one can choose (non-uniquely) a basis
$\{M^{\theta}_k\}_{k=1}^8$ of the subspace $\mathfrak{N}_{\theta}$
consisting of positive operators such that $\sum_{k=1}^8
M^{\theta}_k=I_{\H_A}$ (since the subspace $\mathfrak{N}_{\theta}$
satisfies condition (\ref{L-cond}), see \cite{Sh&Sh}). This basis
can be considered as a quantum observable $\mathcal{M}_{\theta}$. By
using Proposition 1 in \cite{Sh&Sh} and Lemma \ref{trans-l+}
Corollary \ref{sqc-c} can be reformulated in terms of theory of
quantum measurements.\smallskip

\begin{corollary}\label{sqc-c+}
\emph{Let $\,\theta_1,\theta_2,\theta_3$ be positive numbers such
that $\,\theta_1+\theta_2+\theta_3=\pi$. Then all the observables
$\mathcal{M}_{\theta_i}$ and all the bipartite observables
$\mathcal{M}_{\theta_i}\otimes \mathcal{M}_{\theta_j}$ have no
indistinguishable subspaces but the tripartite observable
$\mathcal{M}_{\theta_1}\otimes \mathcal{M}_{\theta_2}\otimes
\mathcal{M}_{\theta_3}$ has indistinguishable subspace spanned by
the vectors (\ref{main-vec+}).}\footnote{We call a subspace $\H_0$
\emph{indistinguishable} for an observable $\mathcal{M}$ if
application of $\mathcal{M}$ to all states supported by $\H_0$ leads
to the same outcomes probability distribution.}
\end{corollary}\smallskip

Note also that Theorem 1 implies the following example of
superactivation of 2-shot  quantum zero-error capacity. \smallskip

\begin{corollary}\label{sqc-c+} \emph{Let $\,\theta_1,\theta_2$ be positive numbers such that $\,\theta_1+\theta_2=\pi/2$. Then}
$$
\bar{Q}_0\left([\Phi_{\theta_1}\otimes
\Phi_{\theta_2}]^{\otimes2}\right)>0\quad \textit{while}
\quad\bar{Q}_0\left(\Phi_{\theta_1}^{\otimes2}\right)=\bar{Q}_0\left(\Phi_{\theta_2}^{\otimes2}\right)=\bar{Q}_0\left(\Phi_{\theta_1}\otimes
\Phi_{\theta_2}\right)=0.
$$
\end{corollary}

\emph{The proof of Theorem \ref{sqc}.} The equality
$\,\bar{Q}_0(\Phi_{\theta})=0\,$ for $\,\theta\neq\pi,\,$ the
inequality $\,\bar{Q}_0(\Phi_{\pi})\leq1\,$ and assertion
$\mathrm{C_n}$  follows from Theorem 1 in \cite{NBC}, since the
inclusion $\mathfrak{L}_{\theta}\subset\mathfrak{N}_{\theta}$
implies
$\,\bar{Q}_0(\Phi_{\theta_1}\otimes\ldots\otimes\Phi_{\theta_n})\leq\bar{Q}_0(\Psi_{\theta_1}\otimes\ldots\otimes\Psi_{\theta_n})\,$
for any channels
$\Psi_{\theta_1}\in\widehat{\mathfrak{L}}_{\theta_1}$,...,$\Psi_{\theta_n}\in\widehat{\mathfrak{L}}_{\theta_n}$.

To prove that $\,\bar{Q}_0(\Phi_{\pi})\geq1\,$ it suffices to show,
by using Lemma \ref{trans-l+}, that the vectors
$|\varphi\rangle=[1,\mathrm{i},0,0]^{\top}$ and
$|\psi\rangle=[0,0,1,\mathrm{i}]^{\top}$ generate error-correcting
code for the channel $\Phi_{\pi}$.\medskip

B) Let
$M_1\in\mathfrak{N}_{\theta_1},\ldots,M_n\in\mathfrak{N}_{\theta_n}$
be arbitrary, $X=M_1\otimes \ldots\otimes M_n$ and $\,\varphi,\,
\psi\,$ be the vectors defined in (\ref{main-vec}). By Lemma
\ref{trans-l+} it suffices to show  that
\begin{equation}\label{mps-one}
   \langle\psi | X | \varphi\rangle=0\quad\textrm{and}\quad
   \langle\psi | X |\psi\rangle=\langle\varphi | X | \varphi\rangle.
\end{equation}
Let $a_k,b_k,...,h_k$ be elements of the matrix $M_k$. We have
$$
\begin{array}{c}
2\langle\psi | X | \varphi\rangle= \langle 3\ldots 3| X|1\ldots
1\rangle+\mathrm{i}\langle 3\ldots 3| X|2\ldots
2\rangle-\mathrm{i}\langle 4\ldots 4| X|1\ldots 1\rangle\\\\+\langle
4\ldots 4|X|2\ldots 2\rangle= g_1 \ldots g_n(1+\gamma_1
\ldots\gamma_n)+h_1 \ldots h_n(\mathrm{i}-\mathrm{i})=0,
\end{array}
$$
since $\gamma_1 \ldots\gamma_n=-1\,$ by the condition
$\,\theta_1+\ldots+\theta_n=\pi(\mathrm{mod}\,2\pi)$,
$$
\begin{array}{c}
2\langle\varphi | X | \varphi\rangle= \langle 1\ldots 1| X|1\ldots
1\rangle+\mathrm{i}\langle 1\ldots 1| X|2\ldots
2\rangle-\mathrm{i}\langle 2\ldots 2| X|1\ldots 1\rangle\\\\+\langle
2\ldots 2|X|2\ldots 2\rangle=a_1\ldots a_n+\mathrm{i}\,(b_1\ldots
b_n-c_1\ldots c_n)+d_1\ldots d_n
\end{array}
$$
and
$$
\begin{array}{c}
2\langle\psi | X | \psi\rangle= \langle 3\ldots 3| X|3\ldots
3\rangle+\mathrm{i}\langle 3\ldots 3| X|4\ldots
4\rangle-\mathrm{i}\langle 4\ldots 4| X|3\ldots 3\rangle\\\\+\langle
4\ldots 4|X|4\ldots 4\rangle=a_1\ldots a_n+\mathrm{i}\,(b_1\ldots
b_n-c_1\ldots c_n)+d_1\ldots d_n.
\end{array}
$$
Thus the both equalities in (\ref{mps-one}) are valid.
\medskip

$\mathrm{C_2}$) To prove this assertion we have to show that the
subspace $\mathfrak{N}_{\theta_1}\otimes\mathfrak{N}_{\theta_2}$
does not satisfy condition (\ref{l-3-c}) if
$\,|\theta_1|+|\theta_2|<\pi$. In the case $\,\theta_1=\theta_2=0$
this follows from assertion $\mathrm{C_n}$. So, we may assume, by
symmetry, that $\theta_2\neq0$.
\smallskip

Throughout the proof we will use the isomorphism
\begin{equation*}
\mathbb{C}^n\otimes\mathbb{C}^m\ni x\otimes y \;\leftrightarrow\;
[x_1y,\ldots,x_ny]^{\top}\in\underbrace{\mathbb{C}^m\oplus\ldots\oplus\mathbb{C}^m}_n
\end{equation*}
and the corresponding isomorphism
\begin{equation}\label{iso}
    \mathfrak{M}_n\otimes\mathfrak{M}_m\ni A\otimes B \;\leftrightarrow\;
[a_{ij}B]\in\mathfrak{M}_{nm}.
\end{equation}

Let $U_1,U_2,V_1,V_2$ be the unitary operators in $\mathbb{C}^2$
determined (in the canonical basis) by the matrices
$$U_1
=\left[\begin{array}{rr}
1 &  0 \\
0  &  \gamma_1
\end{array}\right],\quad
V_1 =\left[\begin{array}{rr}
1 &  0 \\
0  &  \gamma_2
\end{array}\right],\quad
U_2=V_2=\left[\begin{array}{rr}
0 &  1 \\
1  &  0
\end{array}\right].
$$

We will identify $\mathbb{C}^4$ with
$\mathbb{C}^2\oplus\mathbb{C}^2$. So, arbitrary matrices
$M_1\in\mathfrak{N}_{\theta_1}$ and $M_2\in\mathfrak{N}_{\theta_2}$
can be represented as follows
$$
M_1=\left[\begin{array}{cc} A_1 &
 e_1U^*_1+f_1U^*_2\\ g_1U_1+h_1U_2 & A_1\end{array}\right],\;M_2=\left[\begin{array}{cc} A_2 &
 e_2V^*_1+f_2V^*_2\\ g_2V_1+h_2V_2 & A_2\end{array}\right]
$$
or, according to (\ref{iso}), as
$$
M_1=I_2\otimes A_1 +|2\rangle \langle
1|\otimes[g_1U_1+h_1U_2]+|1\rangle\langle
2|\otimes[e_1U^*_1+f_1U^*_2]
$$
and
$$
M_2=I_2\otimes A_2 +|2\rangle \langle
1|\otimes[g_2V_1+h_2V_2]+|1\rangle\langle
2|\otimes[e_2V^*_1+f_2V^*_2],
$$
where $A_1$ and $A_2$ are arbitrary matrices in $\mathfrak{M}_2$.

Assume the existence of orthogonal unit vectors $\varphi$ and $\psi$
in $\mathbb{C}^4\otimes\mathbb{C}^4$ such that
\begin{equation}\label{b-eq}
\langle\psi| M_1\otimes M_2|\varphi\rangle=0\quad\textup{and}\quad
\langle\psi| M_1\otimes M_2|\psi\rangle=\langle\varphi| M_1\otimes
M_2|\varphi\rangle
\end{equation}
for all $M_1\in\mathfrak{N}_{\theta_1}$ and
$M_2\in\mathfrak{N}_{\theta_2}$.

By using the above representations of $M_1$ and $M_2$ we have
$$
\begin{array}{c}
M_1\otimes M_2=[I_2\otimes I_2]\otimes[A_1\otimes A_2]+[I_2\otimes
|2\rangle \langle 1|]\otimes[A_1\otimes [g_2V_1+h_2V_2]]+ \\\\\!
[I_2\otimes |1\rangle \langle 2|]\otimes[A_1\otimes
[e_2V^*_1+f_2V^*_2]]+[|2\rangle \langle 1|\otimes
I_2]\otimes[[g_1U_1+h_1U_2]\otimes A_2]+...
\end{array}
$$

Since $\mathfrak{M}_2\otimes\mathfrak{M}_2=\mathfrak{M}_4$, by
choosing $e_i=f_i=g_i=h_i=0$, $i=1,2$, we obtain from (\ref{b-eq})
that
\begin{equation*}
\langle\psi| I_4\otimes A|\varphi\rangle=0\quad\textup{and}\quad
\langle\psi|I_4\otimes A |\psi\rangle=\langle\varphi| I_4\otimes A
|\varphi\rangle\quad \forall A\in\mathfrak{M}_4.
\end{equation*}
According to (\ref{iso}) we have
\begin{equation*}
 I_4\otimes A=\left[\begin{array}{cccc}
A &  0 & 0 &  0\\
0 &  A & 0 &  0\\
0 &  0 & A &  0\\
0 &  0 & 0 &  A
\end{array}\right]\!,\quad |\varphi\rangle=
\left[\begin{array}{c}
x_1 \\
x_2\\
x_3\\
x_4
\end{array}\right]\!,\quad
|\psi\rangle= \left[\begin{array}{c}
y_1 \\
y_2\\
y_3\\
y_4
\end{array}\right]\!,
\end{equation*}
where $x_i, y_i$ are vectors in $\mathbb{C}^4$. So, the above
relations can be written as the following ones
\begin{equation*}
\sum_{i=1}^4\langle y_i|A|x_i\rangle=0\quad \textup{and} \quad
\sum_{i=1}^4\langle y_i|A|y_i\rangle=\sum_{i=1}^4\langle
x_i|A|x_i\rangle\quad\forall A\in\mathfrak{M}_4
\end{equation*}
which are equivalent to the  operator equalities
\begin{equation}\label{b-eq-2}
\sum_{i=1}^4|y_i\rangle\langle x_i|=0.
\end{equation}
and
\begin{equation}\label{b-eq-1}
\sum_{i=1}^4|y_i\rangle\langle y_i|=\sum_{i=1}^4|x_i\rangle\langle
x_i|
\end{equation}

By choosing $e_i=f_i=g_1=h_1=0$, $i=1,2$, $A_2=0$, $(g_2,h_2)=(1,0)$
and $(g_2,h_2)=(0,1)$ we obtain from (\ref{b-eq}) that
$$
\langle\psi| [I_2\otimes |2\rangle \langle 1|]\otimes[A_1\otimes
V_k]|\varphi\rangle=0
$$
and
$$
\langle\psi|[I_2\otimes |2\rangle \langle 1|]\otimes[A_1\otimes V_k]
|\psi\rangle=\langle\varphi|[I_2\otimes |2\rangle \langle
1|]\otimes[A_1\otimes V_k]|\varphi\rangle
$$
for all $A_1$ in $\mathfrak{M}_2$ and $k=1,2$. According to
(\ref{iso}) we have
$$
[I_2\otimes |2\rangle \langle 1|]\otimes[A_1\otimes
V_k]=\left[\begin{array}{cccc}
0 &  0 & 0 &  0\\
A_1\otimes V_k &  0 & 0 & 0\\
0 &  0 & 0 &  0\\
0 &  0 & A_1\otimes V_k &  0
\end{array}\right]
$$
and hence the above equalities imply
\begin{equation}\label{b-eq-3}
\langle y_2|A\otimes V_k|x_1\rangle+\langle y_4|A\otimes
V_k|x_3\rangle=0\quad\forall A\in\mathfrak{M}_2, \; k=1,2,
\end{equation}
and
\begin{equation}\label{b-eq-4}
\begin{array}{l}
\langle y_2|A\otimes V_k|y_1\rangle+\langle y_4|A\otimes
V_k|y_3\rangle=\\\\\langle x_2|A\otimes V_k|x_1\rangle+\langle
x_4|A\otimes V_k|x_3\rangle\qquad\;\forall A\in\mathfrak{M}_2, \;
k=1,2.
\end{array}
\end{equation}

Similarly, by choosing $e_i=f_i=g_2=h_2=0$, $i=1,2$, $A_1=0$,
$(g_1,h_1)=(1,0)$ and $(g_1,h_1)=(0,1)$ we obtain from (\ref{b-eq})
the equalities
\begin{equation}\label{b-eq-5}
\langle y_3|U_k\otimes A|x_1\rangle+\langle y_4|U_k\otimes
A|x_2\rangle=0, \!\quad\forall  A\in\mathfrak{M}_2, \; k=1,2,
\end{equation}
and
\begin{equation}\label{b-eq-6}
\begin{array}{l}
\langle y_3|U_k\otimes A|y_1\rangle+\langle y_4|U_k\otimes
A|y_2\rangle=\\\\ \langle x_3|U_k\otimes A|x_1\rangle+\langle
x_4|U_k\otimes A|x_2\rangle,\qquad\; \forall A\in\mathfrak{M}_2,\;
k=1,2.
\end{array}
\end{equation}

By the symmetry of condition (\ref{b-eq}) with respect to $\varphi$
and $\psi$ relations (\ref{b-eq-3}) and (\ref{b-eq-5}) imply
respectively
\begin{equation}\label{b-eq-3+}
\langle x_2|A\otimes V_k|y_1\rangle+\langle x_4|A\otimes
V_k|y_3\rangle=0\quad\forall A\in\mathfrak{M}_2, \; k=1,2,
\end{equation}
and
\begin{equation}\label{b-eq-5+}
\langle x_3|U_k\otimes A|y_1\rangle+\langle x_4|U_k\otimes
A|y_2\rangle=0, \!\quad\forall  A\in\mathfrak{M}_2, \; k=1,2.
\end{equation}

Finally, by choosing $A_1=A_2=0$ and appropriate values of
$e_i,f_i,g_i,h_i$, $i=1,2$, one can obtain from (\ref{b-eq}) the
following equalities
\begin{equation}\label{b-eq-7}
\langle y_4|U_k\otimes V_l\,|x_1\rangle=\langle x_4|U_k\otimes
V_l\,|y_1\rangle=0\qquad\;\; k,l=1,2,
\end{equation}
\begin{equation}\label{b-eq-8}
\langle y_4|U_k\otimes V_l\,|y_1\rangle\,=\langle x_4|U_k\otimes
V_l\,|x_1\rangle,\qquad\qquad k,l=1,2,
\end{equation}
\begin{equation}\label{b-eq-9}
\langle y_3|U_k\otimes V^*_l|x_2\rangle=\langle x_3|U_k\otimes
V^*_l|y_2\rangle=0\;\,\,\qquad k,l=1,2,
\end{equation}
\begin{equation}\label{b-eq-10}
\langle y_3|U_k\otimes V^*_l|y_2\rangle=\langle x_3|U_k\otimes
V^*_l|x_2\rangle,\qquad\qquad k,l=1,2.
\end{equation}

We will prove below that the system (\ref{b-eq-2})-(\ref{b-eq-10})
has no nontrivial solutions.\medskip

We will use the following lemmas. \smallskip

\begin{lemma}\label{dl}
A) \emph{Equations (\ref{b-eq-2}) and  (\ref{b-eq-1}) imply that all
the vectors $\,x_i, y_i$, $i=\overline{1,4}$, lie in some 2-D
subspace of $\,\mathbb{C}^4$.}

B) \emph{If $\,x_{i_0}=y_{i_0}=0\,$ for some $i_0$ then equations
(\ref{b-eq-2}) and  (\ref{b-eq-1}) imply that all the vectors
$\,x_i, y_i$, $i=\overline{1,4}$, are collinear.}
\end{lemma}

\begin{proof} A) Consider the  $4\times4$ - matrices
$$
X=[\langle x_i|x_j\rangle],\quad Y=[\langle y_i|y_j\rangle],\quad
Z=[\langle x_i|y_j\rangle].
$$
It is easy to see that (\ref{b-eq-2}) implies $XY=0$ while
(\ref{b-eq-1}) shows that $X^2=ZZ^*$  and $Y^2=Z^*Z$. Hence
$\,\mathrm{rank}X=\mathrm{rank}Y\leq2$.

Since (\ref{b-eq-1}) implies that the sets $\{x_i\}_{i=1}^4$ and
$\{y_i\}_{i=1}^4$ have the same linear hull, the above inequality
shows that this linear hull has dimension $\leq2$.\smallskip

B) This assertion is proved similarly, since the same argumentation
with $3\times3$ - matrices $X,Y,Z$ implies
$\,\mathrm{rank}X=\mathrm{rank}Y\leq1$.
\end{proof}

\begin{lemma}\label{sys-l-1} A) \emph{The  condition
\begin{equation}\label{b-eq-7+}
\langle z_4|U_k\otimes V_l|z_1\rangle=0\qquad\;\;\, k,l=1,2,\quad
\end{equation}
holds if and only if the pair $(z_1,z_4)$ has one of the following
forms:
$$
\!\!\!\begin{array}{l} 1)\; z_1=\left[\begin{array}{c}
\!\mu_1\! \\
\!s\!
\end{array}\right]\!\otimes\!\left[\begin{array}{c}
\!a\! \\
\!b\!
\end{array}\right]\!, \; z_4=\left[\begin{array}{c}
\!\bar{\mu}_1\! \\
\!\!-s\!\!
\end{array}\right]\!\otimes\!\left[\begin{array}{c}
\!c\! \\
\!d\!
\end{array}\right]\!;\\\\
2)\;z_1=\left[\begin{array}{c}
\!a\! \\
\!b\!
\end{array}\right]\!\otimes\!\left[\begin{array}{c}
\!\mu_2\! \\
\!s\!
\end{array}\right]\!, \; z_4=\left[\begin{array}{c}
\!c\! \\
\!d\!
\end{array}\right]\!\otimes\!\left[\begin{array}{c}
\!\bar{\mu}_2\! \\
\!\!-s\!\!
\end{array}\right]\!;\\\\
3)\;z_1=a\left[\begin{array}{c}
\!\mu_1\! \\
\!1\!
\end{array}\right]\!\otimes\!\left[\begin{array}{c}
\!\mu_2\! \\
\!s\!
\end{array}\right]+
b\left[\begin{array}{c}
\!\mu_1\! \\
\!\!-1\!\!
\end{array}\right]\!\otimes\!\left[\begin{array}{c}
\!\mu_2\! \\
\!\!-s\!\!
\end{array}\right]\!,
\; z_4=c\left[\begin{array}{c}
\!\bar{\mu}_1\! \\
\!1\!
\end{array}\right]\!\otimes\!\left[\begin{array}{c}
\!\bar{\mu}_2\! \\
\!\!-s\!
\end{array}\right]+
d\left[\begin{array}{c}
\!\bar{\mu}_1\! \\
\!\!-1\!\!
\end{array}\right]\!\otimes\!\left[\begin{array}{c}
\!\bar{\mu}_2\! \\
\!s\!
\end{array}\right]\!;\\\\
4)\;z_1=h\left[\begin{array}{c}
\!\mu_1\! \\
\!s\!
\end{array}\right]\!\otimes\!\left[\begin{array}{c}
\!\mu_2\! \\
\!t\!
\end{array}\right]\!,
\; z_4=\left[\begin{array}{c}
\!\bar{\mu}_1\! \\
\!\!-s\!\!
\end{array}\right]\!\otimes\!\left[\begin{array}{c}
\!a\! \\
\!b\!
\end{array}\right]+
\left[\begin{array}{c}
\!c\! \\
\!d\!
\end{array}\right]\!\otimes\!\left[\begin{array}{c}
\!\bar{\mu}_2\! \\
\!\!-t\!\!
\end{array}\right]\!;\\\\
5)\;z_1=\left[\begin{array}{c}
\!\mu_1\! \\
\!\!-s\!\!
\end{array}\right]\!\otimes\!\left[\begin{array}{c}
\!a\! \\
\!b\!
\end{array}\right]+
\left[\begin{array}{c}
\!c\! \\
\!d\!
\end{array}\right]\!\otimes\!\left[\begin{array}{c}
\!\mu_2\! \\
\!\!-t\!\!
\end{array}\right]\!,
\; z_4=h\left[\begin{array}{c}
\!\bar{\mu}_1\! \\
\!s\!
\end{array}\right]\!\otimes\!\left[\begin{array}{c}
\!\bar{\mu}_2\! \\
\!t\!
\end{array}\right]\!;
\end{array}
$$
where $\,\mu_k=\sqrt{\gamma_k},\,k=1,2,$ and $\,a,b,c,d,h \in
\mathbb{C}$, $s=\pm1$, $t=\pm1$.}\medskip

B) \emph{Validity of (\ref{b-eq-7}) and (\ref{b-eq-8}) for vectors
$\,x_i,y_i$, $i=1,4$, implies}
$$
\langle y_4|U_k\otimes V_l|y_1\rangle=\langle x_4|U_k\otimes
V_l|x_1\rangle=0.
$$
\end{lemma}

\begin{lemma}\label{sys-l-2}
A) \emph{The  condition
\begin{equation}\label{b-eq-9+}
\langle z_3|U_k\otimes V^*_l|z_2\rangle=0\qquad\;\;\, k,l=1,2,\quad
\end{equation}
holds if and only if the pair $(z_2,z_3)$ has one of the following
forms:
$$
\!\!\!\begin{array}{l} 1)\; z_2=\left[\begin{array}{c}
\!\mu_1\! \\
\!s\!
\end{array}\right]\!\otimes\!\left[\begin{array}{c}
\!a\! \\
\!b\!
\end{array}\right]\!, \; z_3=\left[\begin{array}{c}
\!\bar{\mu}_1\! \\
\!\!-s\!\!
\end{array}\right]\!\otimes\!\left[\begin{array}{c}
\!c\! \\
\!d\!
\end{array}\right]\!;\\\\
2)\;z_2=\left[\begin{array}{c}
\!a\! \\
\!b\!
\end{array}\right]\!\otimes\!\left[\begin{array}{c}
\!\bar{\mu}_2\! \\
\!s\!
\end{array}\right]\!, \; z_3=\left[\begin{array}{c}
\!c\! \\
\!d\!
\end{array}\right]\!\otimes\!\left[\begin{array}{c}
\!\mu_2\! \\
\!\!-s\!\!
\end{array}\right]\!;\\\\
3)\;z_2=a\left[\begin{array}{c}
\!\mu_1\! \\
\!1\!
\end{array}\right]\!\otimes\!\left[\begin{array}{c}
\!\bar{\mu}_2\! \\
\!s\!
\end{array}\right]+
b\left[\begin{array}{c}
\!\mu_1\! \\
\!\!-1\!\!
\end{array}\right]\!\otimes\!\left[\begin{array}{c}
\!\bar{\mu}_2\! \\
\!\!-s\!\!
\end{array}\right]\!,
\; z_3=c\left[\begin{array}{c}
\!\bar{\mu}_1\! \\
\!\!-1\!\!
\end{array}\right]\!\otimes\!\left[\begin{array}{c}
\!\mu_2\! \\
\!s\!
\end{array}\right]+
d\left[\begin{array}{c}
\!\bar{\mu}_1\! \\
\!1\!
\end{array}\right]\!\otimes\!\left[\begin{array}{c}
\!\mu_2\! \\
\!\!-s\!\!
\end{array}\right]\!;\\\\
4)\;z_2=h\left[\begin{array}{c}
\!\mu_1\! \\
\!s\!
\end{array}\right]\!\otimes\!\left[\begin{array}{c}
\!\bar{\mu}_2\! \\
\!t\!
\end{array}\right]\!,
\; z_3=\left[\begin{array}{c}
\!\bar{\mu}_1\! \\
\!\!-s\!\!
\end{array}\right]\!\otimes\!\left[\begin{array}{c}
\!a\! \\
\!b\!
\end{array}\right]+
\left[\begin{array}{c}
\!c\! \\
\!d\!
\end{array}\right]\!\otimes\!\left[\begin{array}{c}
\!\mu_2\! \\
\!\!-t\!\!
\end{array}\right]\!;\\\\
5)\;z_2=\left[\begin{array}{c}
\!\mu_1\! \\
\!\!-s\!\!
\end{array}\right]\!\otimes\!\left[\begin{array}{c}
\!a\! \\
\!b\!
\end{array}\right]+
\left[\begin{array}{c}
\!c\! \\
\!d\!
\end{array}\right]\!\otimes\!\left[\begin{array}{c}
\!\bar{\mu}_2\! \\
\!\!-t\!\!
\end{array}\right]\!,
\; z_3=h\left[\begin{array}{c}
\!\bar{\mu}_1\! \\
\!s\!
\end{array}\right]\!\otimes\!\left[\begin{array}{c}
\!\mu_2\! \\
\!t\!
\end{array}\right]\!;
\end{array}
$$
where $\,\mu_k=\sqrt{\gamma_k},\,k=1,2,$ and $\,a,b,c,d,h \in
\mathbb{C}$, $s=\pm1$, $t=\pm1$.}\medskip

B) \emph{Validity of (\ref{b-eq-9}) and (\ref{b-eq-10}) for vectors
$\,x_i,y_i$, $i=2,3$, implies}
$$
\langle y_3|U_k\otimes V^*_l|y_2\rangle=\langle x_3|U_k\otimes
V^*_l|x_2\rangle=0.
$$
\end{lemma}
\medskip

Lemmas \ref{sys-l-1} and \ref{sys-l-2} are proved in the Appendix.
\medskip

\begin{lemma}\label{sl}  \emph{Let $\;|\theta_1|+|\theta_2|<\pi\,$
then}
\begin{enumerate}[A)]
    \item \emph{$\langle x|U_1 |x\rangle\neq 0\,$ and $\,\langle x|V_1 |x\rangle\neq 0\,$ for any
nonzero $x\in\mathbb{C}^2$;}
    \item \emph{$ \langle y|U_1\otimes V_1| y\rangle\neq 0\,$ and $\,\langle y|U_1\otimes V^*_1| y\rangle\neq 0\,$ for any nonzero  $\,y\in
\mathbb{C}^2\otimes\mathbb{C}^2$.}
\end{enumerate}
\end{lemma}

\begin{proof}

A) Let $\,|x\rangle=[x_1,x_2]^{\top}\neq0\;$ then $\langle
x|U_1|x\rangle=|x_1|^2+\gamma_1|x_2|^2\neq0$ and $\langle
x|V_1|x\rangle=|x_1|^2+\gamma_2|x_2|^2\neq0$, since
$\theta_1,\theta_2\neq\pi$.
\smallskip

B) Since $U_1\otimes V_1=\mathrm{diag}\{1, \gamma_2, \gamma_1,
\gamma_1\gamma_2\}$, the equality $\langle y|U_1\otimes V_1|
y\rangle=0$ for vector $|y\rangle=[y_1,y_2,y_3,y_4]^{\top}$ means
that
$|y_1|^2+|y_2|^2\gamma_2+|y_3|^2\gamma_1+|y_4|^2\gamma_1\gamma_2=0$.
By the condition $\;|\theta_1|+|\theta_2|<\pi\,$ the numbers $0,1,
\gamma_2, \gamma_1, \gamma_1\gamma_2$ are extreme points of a convex
polygon in complex plane, so the last equality can be valid only if
$\,y_i=0$ for all $i$.\smallskip

Similarly one can show that $\langle y|U_1\otimes V^*_1| y\rangle=0$
implies $\,y=0$.
\end{proof}

\begin{lemma}\label{rl}
\emph{Let $\,p$ and  $\,q\,$ be complex numbers such that
$\;|p|^2+|q|^2=1$. If $\{|x_i\rangle\}_{i=1}^4$ and
$\{|y_i\rangle\}_{i=1}^4$ satisfy the system
(\ref{b-eq-2})-(\ref{b-eq-10}) then $\{|p x_i-q
y_i\rangle\}_{i=1}^4$ and $\{|\bar{q} x_i+\bar{p}
y_i\rangle\}_{i=1}^4$ also satisfy the system
(\ref{b-eq-2})-(\ref{b-eq-10}).}
\end{lemma}

\begin{proof} It suffices to note that the condition
\begin{equation*}
\langle\varphi|A|\psi\rangle=\langle\psi|A|\varphi\rangle=\langle\psi|A
|\psi\rangle-\langle\varphi|A|\varphi\rangle=0
\end{equation*} is
invariant under the rotation $|\varphi\rangle\mapsto
p|\varphi\rangle-q |\psi\rangle,
\;|\psi\rangle\mapsto\bar{q}|\varphi\rangle+\bar{p}|\psi\rangle$.
\end{proof}

\begin{lemma}\label{pl}
\emph{If $\;|\theta_1|+|\theta_2|<\pi\,$ then the system
(\ref{b-eq-2})-(\ref{b-eq-10}) has no nontrivial solution of the
form  $\,|x_i\rangle=\alpha_i|z\rangle$ and
$\,|y_i\rangle=\beta_i|z\rangle$, $i=\overline{1,4}$.}
\end{lemma}

\begin{proof} Assume that $\,|x_i\rangle=\alpha_i|z\rangle$ and
$\,|y_i\rangle=\beta_i|z\rangle$, $i=\overline{1,4},$ form a
nontrivial solution of the system (\ref{b-eq-2})-(\ref{b-eq-10}).
Then (\ref{b-eq-2}) implies that
$|\alpha\rangle=[\alpha_1,\ldots,\alpha_4]^{\top}$ and
$|\beta\rangle=[\beta_1,\ldots,\beta_4]^{\top}$ are orthogonal
nonzero vectors of the same norm. By Lemma \ref{sl}B it follows from
(\ref{b-eq-7}) and (\ref{b-eq-9}) that
\begin{equation}\label{vnr}
\alpha_1\beta_4=\alpha_4\beta_1=\alpha_2\beta_3=\alpha_3\beta_2=0.
\end{equation}
If $|\alpha_i|+|\beta_i|>0$ for all $i$ then we may consider by
Lemma \ref{rl} that $\alpha_i\neq0$ for all $i$ and hence
(\ref{vnr}) implies $\beta_i=0$ for all $i$, i.e. $|\beta\rangle=0$.

Assume that $\alpha_1=\beta_1=0\,$ and consider the following two
cases.

1) If $\alpha_4=\beta_4=0$ then the condition
$\langle\beta|\alpha\rangle=0$ implies $|\alpha_i|+|\beta_i|>0$ for
$i=2,3$ and we may consider by Lemma \ref{rl} that $\alpha_i\neq0$
for $i=2,3$. Hence (\ref{vnr}) implies $\beta_i=0$ for $i=2,3$ and
hence $|\beta\rangle=0$.

2) If $|\alpha_4|+|\beta_4|>0$ then we may consider by Lemma
\ref{rl} that $\alpha_4\neq0$. By Lemma \ref{sl}B it follows from
(\ref{b-eq-3+}) with $A=U_1$ and (\ref{b-eq-5+}) with $A=V_1$ that
$\alpha_4\beta_2=\alpha_4\beta_3=0$, which implies
$\beta_2=\beta_3=0$. Hence the condition
$\langle\beta|\alpha\rangle=0$ can be valid only if
$|\beta\rangle=0$.

By the similar way one can show that the assumption
$\alpha_i=\beta_i=0$ for $i=2,3,4\,$ leads to a contradiction.
\end{proof}

Assume now that the collections $\{x_i\}_1^4$ and $\{y_i\}_1^4$ form
a nontrivial solution of the system (\ref{b-eq-2})-(\ref{b-eq-10}).

If $\,x_i\nparallel y_i\,$ for some $\,i\,$  then
(\ref{b-eq-7})-(\ref{b-eq-10}) and Lemmas \ref{sys-l-1}B,
\ref{sys-l-2}B imply
$$
\langle y_{5-i}|W_i| x_i\rangle=\langle x_{5-i}|W_i|
y_i\rangle=\langle x_{5-i}|W_i| x_i\rangle=\langle y_{5-i}|W_i|
y_i\rangle=0,
$$
where $W_1=U_1\otimes V_1$, $W_2=U_1\otimes V^*_1$,
$W_3=U^*_1\otimes V_1$, $W_4=U^*_1\otimes V^*_1$. By Lemma \ref{dl}A
we have $x_{5-i},y_{5-i}\in\mathrm{lin}\{x_i,y_i\}$, so the above
equalities show that $\langle x_{5-i}|W_i| x_{5-i}\rangle=\langle
y_{5-i}|W_i| y_{5-i}\rangle=0$. Lemma \ref{sl}B implies
$x_{5-i}=y_{5-i}=0$, which contradicts to the assumption
$x_i\nparallel y_i$ by Lemma \ref{dl}B. \smallskip

Thus, $\,x_i\parallel y_i\,$ for all $\,i=\overline{1,4}$. By Lemma
\ref{pl} we may assume in what follows that
\begin{equation}\label{assum}
|x_i\rangle=\alpha_i|z_i\rangle \;\,\textup{and}\;\,
|y_i\rangle=\beta_i|z_i\rangle,\; \textup{where} \;\,|z_i\rangle\,
\;\textup{are non-collinear vectors}\footnote{In the sense that all
the vectors $\,|z_i\rangle$, $i=\overline{1,4},\,$ are not collinear
to each other, but some two of them  may be collinear.}.
\end{equation}

Then Lemma \ref{dl}B implies
$|\alpha_i|+|\beta_i|>0,\;i=\overline{1,4},\,$ and equations
(\ref{b-eq-2}), (\ref{b-eq-1}) can be rewritten as follows
\begin{equation}\label{eq-2-sf}
\sum_{i=1}^4 \bar{\beta}_i \alpha_i|z_i\rangle\langle z_i|=0,
\end{equation}
\begin{equation}\label{eq-1-sf}
\sum_{i=1}^4\left[|\beta_i|^2-|\alpha_i|^2\right]|z_i\rangle\langle
z_i|=0.
\end{equation}

By Lemma \ref{rl} we may assume that $\beta_1=0$ and hence
$\alpha_1\neq0$. There are two cases:

1) If $\,\beta_i\alpha_i\neq0\,$ for all $\,i>1\,$ then
(\ref{eq-2-sf}) and Lemma \ref{al-1} in the Appendix imply
$z_2\parallel z_3\parallel z_4$. Then it follows from
(\ref{eq-1-sf}) that
$$
|\alpha_1|^2|z_1\rangle\langle z_1|+[\ldots]|z_2\rangle\langle
z_2|=0
$$
and hence $z_1\parallel z_2\parallel z_3\parallel z_4$ contradicting
to the assumption (\ref{assum}).

2) If there is $\,k>1\,$ such that $\,\beta_k\alpha_k=0\,$ then
(\ref{eq-2-sf}) implies that either $\,\beta_i\alpha_i\neq0\,$ and
$\,\beta_j\alpha_j\neq0\,$ or $\beta_i\alpha_i=\beta_j\alpha_j=0$,
where $i$ and $j>i$ are complementary indexes to $1$ and $k$.

If $\,\beta_i\alpha_i\neq0\,$ and $\,\beta_j\alpha_j\neq0\,$ then it
follows from (\ref{eq-2-sf}) that $z_i\parallel z_j$ and
(\ref{eq-1-sf}) implies
$$
|\alpha_1|^2|z_1\rangle\langle z_1|+p|z_k\rangle\langle
z_k|+[\ldots]|z_i\rangle\langle z_i|=0,
$$
where $p$ is a nonzero number (equal either to $|\alpha_k|^2$ or to
$-|\beta_k|^2$). Hence $z_1\parallel z_k$ by Lemma \ref{al-1} in the
Appendix.

Thus $z_1\parallel z_k$ and $z_i\parallel z_j$. By Lemma \ref{sl}B
it follows from (\ref{b-eq-7}) and (\ref{b-eq-9}) that $k\neq4$ and
$(i,j)\neq(2,3)$. So, we have only two possibilities:

a) $k=2, i=3, j=4$. In this case $z_3\parallel z_4$ and
(\ref{b-eq-3+})  with $A=U_1$ implies
$$
\bar{\alpha}_4\beta_3\langle z_4|U_1\otimes
V_1|z_3\rangle=-\bar{\alpha}_2\beta_1\langle z_2|U_1\otimes
V_1|z_1\rangle=0\quad (\textrm{since}\;\beta_1=0).
$$
Hence Lemma \ref{sl} shows that $\,\alpha_4\beta_3=0\,$
contradicting to the assumption.

b) $k=3, i=2, j=4$. In this case $z_2\parallel z_4$ and
(\ref{b-eq-5+}) with $A=V_1$ implies
$$
\bar{\alpha}_4\beta_2\langle z_4|U_1\otimes
V_1|z_2\rangle=-\bar{\alpha}_3\beta_1\langle z_3|U_1\otimes
V_1|z_1\rangle=0\quad (\textrm{since}\;\beta_1=0).
$$
Hence Lemma \ref{sl}B shows that $\,\alpha_4\beta_2=0\,$
contradicting to the assumption.\smallskip

So, we necessarily have $\,\beta_i\alpha_i=0\,$ for all
$\,i=\overline{1,4}$. Since the vectors $z_1,\ldots,z_4\,$ are not
collinear by assumption (\ref{assum}), equality (\ref{eq-1-sf}) and
Lemma \ref{dl}B imply that there are two nonzero $\alpha_i$ and two
nonzero $\beta_i$. Thus, we have (up to permutation) the following
cases
$$
\mathrm{a)}\,
|\varphi\rangle,\!|\psi\rangle\!=\!\!\left[\begin{array}{c}
\!x_1\! \\
\!x_2\!\\
0\\
0
\end{array}\right]\!\!,\!
\left[\begin{array}{c}
0 \\
0\\
\!y_3\!\\
\!y_4\!
\end{array}\right]\!\!;\;\,
\mathrm{b)}\,
|\varphi\rangle,\!|\psi\rangle\!=\!\!\left[\begin{array}{c}
\!x_1\! \\
0\\
\!x_3\!\\
0
\end{array}\right]\!\!,\!
\left[\begin{array}{c}
0 \\
\!y_2\!\\
0\\
\!y_4\!
\end{array}\right]\!\!;\;\,
\mathrm{c)}\,
|\varphi\rangle,\!|\psi\rangle\!=\!\!\left[\begin{array}{c}
\!x_1\! \\
0\\
0\\
\!x_4\!
\end{array}\right]\!\!,\!
\left[\begin{array}{c}
0 \\
\!y_2\!\\
\!y_3\!\\
0
\end{array}\right]\!\!,
$$
where $x_1\nparallel x_k$ and $y_i\nparallel y_j$ (since otherwise
(\ref{eq-1-sf}) implies $x_1\parallel x_k\parallel y_i\parallel
y_j$).

Show first that case c) is not possible. It follows from
(\ref{b-eq-3}) with $A=U_1$ and (\ref{b-eq-5}) with $A=V_1$ that
\begin{equation*}
\langle y_2|U_1\otimes V_1|x_1\rangle=\langle y_3|U_1\otimes
V_1|x_1\rangle=0.
\end{equation*}
Since $y_2\nparallel y_3$, Lemma \ref{dl}A shows that $x_1\in
\mathrm{lin}\{y_2,y_3\}$  and the above equalities  imply $\langle
x_1|U_1\otimes V_1|x_1\rangle=0$. By Lemma \ref{sl}B $\,x_1=0$.
\medskip

It is more difficult to show impossibility of cases a) and b). We
will consider these cases simultaneously by denoting
$\,z_2=x_2,z_3=y_3\,$ in case a), $\,z_2=y_2,z_3=x_3\,$ in case b)
and $\,z_1=x_1,z_4=y_4\,$ in the both cases. The system
(\ref{b-eq-2})-(\ref{b-eq-10}) implies the following equations:
\begin{equation}\label{nbe-1}
|x_1\rangle\langle x_1|+|x_i\rangle\langle x_i|=|y_j\rangle\langle
y_j|+|y_4\rangle\langle y_4|
\end{equation}
where $(i,j)=(2,3)$ in case a) and $(i,j)=(3,2)$ in case b),
\begin{equation}\label{nbe-2}
\langle z_3|U_k\otimes A|x_1\rangle=-\sigma_{\!*}\langle
y_4|U_k\otimes A|z_2\rangle\quad \forall A\in\mathfrak{M}_2,\;k=1,2,
\end{equation}
\begin{equation}\label{nbe-3}
\langle z_2|A\otimes V_k|x_1\rangle=+\sigma_{\!*}\langle
y_4|A\otimes V_k|z_3\rangle\quad \forall A\in\mathfrak{M}_2,\;k=1,2,
\end{equation}
where $\sigma_{\!*}=1$ in case a) and $\sigma_{\!*}=-1$ in case b),
\begin{equation}\label{nbe-4}
\langle y_4|U_k\otimes V_l|x_1\rangle=0\quad  k,l=1,2,
\end{equation}
\begin{equation}\label{nbe-5}
\langle z_3|U_k\otimes V^*_l|z_2\rangle=0\quad k,l=1,2.
\end{equation}

It follows from (\ref{nbe-4}) and (\ref{nbe-5}) that the pairs
$(x_1,y_4)$ and $(z_2,z_3)$ must have one of the forms 1-5 presented
in part A of Lemmas \ref{sys-l-1} and \ref{sys-l-2}
correspondingly.\smallskip

Assume first that the both pairs $(x_1,y_4)$ and $(z_2,z_3)$ have
forms 1-2. In this case $\,x_1,z_2,z_3,y_4\,$ are tensor product
vectors. By Lemma \ref{al-2} in the Appendix (\ref{nbe-1}) can be
valid only in the following cases (1-4):

1) $|z_i\rangle=|p\rangle\otimes|a_i\rangle$, $i=\overline{1,4}$. It
follows from (\ref{nbe-2}) that
$$
\langle p|U_1|p\rangle\langle a_3|A|a_1\rangle=-\sigma_{\!*}\langle
p|U_1|p\rangle\langle a_4|A|a_2\rangle\quad \forall
A\in\mathfrak{M}_2.
$$
Since $\langle p|U_1|p\rangle\neq0$ by Lemma \ref{sl}A, we have
$\,a_1\parallel a_2\,$ and $\,a_3\parallel a_4\,$. In case a) this
and (\ref{nbe-1}) implies $\,x_1\parallel x_2\parallel\,y_3\parallel
y_4\,$ contradicting to
 assumption (\ref{assum}). In case b) it means  $\,x_1\parallel
y_2\,$ and $\,x_3\parallel y_4\,$. The assumption $\,x_1\nparallel
x_3\,$ and (\ref{nbe-1}) show that this can be valid only if
$\,|x_1\rangle\langle x_1|=|y_2\rangle\langle y_2|\,$ and
$\,|x_3\rangle\langle x_3|=|y_4\rangle\langle y_4|$. So, this case
is reduced to case 4) considered below.

2) $|z_i\rangle=|a_i\rangle\otimes|p\rangle$, $i=\overline{1,4}$.
Similarly to case 1) this case is reduced to case 4) by using
(\ref{nbe-3}) instead of (\ref{nbe-2}).

3) $|x_1\rangle\langle x_1|=|y_4\rangle\langle y_4|$ and
$|z_2\rangle\langle z_2|=|z_3\rangle\langle z_3|$. This is not
possible due to (\ref{nbe-4})-(\ref{nbe-5}) and Lemma \ref{sl}B.

4) $|x_1\rangle\langle x_1|=|y_i\rangle\langle y_i|$ and
$|x_{5-i}\rangle\langle x_{5-i}|=|y_4\rangle\langle y_4|$, where
$i=3$ in case a) and $i=2$ in case b).

If $\,i=3\,$ then $y_3=\alpha x_1$, $y_4=\beta x_2$,
$|\alpha|=|\beta|=1$, and (\ref{nbe-2}) with $\sigma_*=1$ implies
\begin{equation}\label{ins-eq}
\bar{\alpha}\langle x_1|U_1\otimes A|x_1\rangle=-\bar{\beta}\langle
x_2|U_1\otimes A|x_2\rangle\quad \forall A\in\mathfrak{M}_2.
\end{equation}
Since $x_1$ and $x_2$ are product vectors, it follows from
(\ref{ins-eq}) and Lemma \ref{sl}A that
$$
x_1=a\otimes p\quad \textrm{and}\quad x_2=b\otimes p
$$
for some nonzero vectors $\,a,b,p$. Hence (\ref{nbe-4}),
(\ref{nbe-5}) and Lemma \ref{sl}A imply
\begin{equation*}
\langle b|U_k|a\rangle=\langle b|U^*_k|a\rangle=0,\quad k=1,2.
\end{equation*}
If $\gamma_1\neq1$ (i.e. $\theta_1\neq0$) then this can not be valid
for nonzero vectors $a$ and $b$. If $\gamma_1=1$ then (\ref{ins-eq})
shows that $\bar{\alpha}\|a\|^2=-\bar{\beta}\|b\|^2$ while
(\ref{nbe-3}) with $\sigma_*=1$  and Lemma \ref{sl}A imply
$\bar{\beta}\alpha=1$, i.e. $\alpha=\beta$.
\smallskip

Similarly, if $\,i=2\,$ then by using Lemma \ref{sl}A one can obtain
from (\ref{nbe-3}) that
$$
x_1\parallel y_2\parallel p\otimes a\quad \textrm{and}\quad
x_3\parallel y_4\parallel p\otimes b
$$
for some nonzero vectors $\,a,b,p\,$. Hence (\ref{nbe-4}),
(\ref{nbe-5}) and Lemma \ref{sl}A imply
\begin{equation*}
\langle b|V_k|a\rangle=\langle b|V^*_k|a\rangle=0,\quad k=1,2,
\end{equation*}
which can not be valid for nonzero vectors $a$ and $b$ (since
$\theta_2\neq0\;\Rightarrow\,\gamma_2\neq\bar{\gamma}_2$).\medskip

Assume now that the pair $(x_1,y_4)$ have form 3 in Lemma
\ref{sys-l-1}A, i.e.
$$
x_1=a\left[\begin{array}{c}
\!\mu_1\! \\
\!1\!
\end{array}\right]\otimes\left[\begin{array}{c}
\!\mu_2\! \\
\!s\!
\end{array}\right]+
b\left[\begin{array}{c}
\!\mu_1\! \\
\!\!-1\!
\end{array}\right]\otimes\left[\begin{array}{c}
\!\mu_2\! \\
\!\!-s\!
\end{array}\right]\!,
\; y_4=c\left[\begin{array}{c}
\!\bar{\mu}_1\! \\
\!1\!
\end{array}\right]\otimes\left[\begin{array}{c}
\!\bar{\mu}_2\! \\
\!\!-s\!
\end{array}\right]+
d\left[\begin{array}{c}
\!\bar{\mu}_1\! \\
\!\!-1\!
\end{array}\right]\otimes\left[\begin{array}{c}
\!\bar{\mu}_2\! \\
\!s\!
\end{array}\right]\!,
$$
where $s=\pm1$, and show  incompatibility of the system
(\ref{nbe-1})-(\ref{nbe-5}) if the pair $(z_2,z_3)$ has forms 1,2,3
in Lemma \ref{sys-l-2}A. We will do this by reducing to the case of
tensor product vectors $\,x_1,z_2,z_3,y_4\,$ considered before.

1) The pair $(z_2,z_3)$ has form 1, i.e.
$$
z_2=\left[\begin{array}{c}
\!\mu_1\! \\
\!t\!
\end{array}\right]\otimes\left[\begin{array}{c}
\!p\! \\
\!q\!
\end{array}\right]\!,\; z_3=
\left[\begin{array}{c}
\!\bar{\mu}_1\! \\
\!\!-t\!\!
\end{array}\right]\otimes\left[\begin{array}{c}
\!x\! \\
\!y\!
\end{array}\right]\!,\;\; t=\pm1,\,|p|+|q|\neq0,\,|x|+|y|\neq0.
$$
By substituting the expressions for $x_1,z_2,z_3,y_4$ into
(\ref{nbe-2}) and by noting that
\begin{equation}\label{note}
\left\langle\begin{array}{c}
\!\bar{\mu}_1\! \\
\!s\!
\end{array}\right|U_k
\left|\begin{array}{c}
\!\mu_1\! \\
\!\!-s\!\!
\end{array}\right\rangle=0,\;\; s=\pm1,\;\; k=1,2
\end{equation}
we obtain
$$
b\left\langle\begin{array}{c}
\!\bar{\mu}_1\! \\
\!\!-1\!\!
\end{array}\right|U_k
\left|\begin{array}{c}
\!\mu_1\! \\
\!\!-1\!\!
\end{array}\right\rangle\!\!
\left\langle\begin{array}{c}
\!x\! \\
\!y\!
\end{array}\right|A
\left|\begin{array}{c}
\!\mu_2\! \\
\!\!-s\!\!
\end{array}\right\rangle
= -\sigma_{\!*}\bar{c}\left\langle\begin{array}{c}
\!\bar{\mu}_1\! \\
\!1\!
\end{array}\right|U_k
\left|\begin{array}{c}
\!\mu_1\! \\
\!\!1\!
\end{array}\right\rangle\!\!
\left\langle\begin{array}{c}
\!\bar{\mu}_2\! \\
\!\!-s\!\!
\end{array}\right|A
\left|\begin{array}{c}
\!p\! \\
\!q\!
\end{array}\right\rangle\quad \textup{if}\;\; t=1
$$ and
$$
a\left\langle\begin{array}{c}
\!\bar{\mu}_1\! \\
\!1\!
\end{array}\right|U_k
\left|\begin{array}{c}
\!\mu_1\! \\
\!1\!
\end{array}\right\rangle\!\!
\left\langle\begin{array}{c}
\!x\! \\
\!y\!
\end{array}\right|A
\left|\begin{array}{c}
\!\mu_2\! \\
\!s\!
\end{array}\right\rangle
= -\sigma_{\!*}\bar{d}\left\langle\begin{array}{c}
\!\bar{\mu}_1\! \\
\!\!-1\!\!
\end{array}\right|U_k
\left|\begin{array}{c}
\!\mu_1\! \\
\!\!-1\!\!
\end{array}\right\rangle\!\!
\left\langle\begin{array}{c}
\!\bar{\mu}_2\! \\
\!\!s\!
\end{array}\right|A
\left|\begin{array}{c}
\!p\! \\
\!q\!
\end{array}\right\rangle\quad \textup{if}\;\; t=-1.
$$
Validity of this equality for all $A\in\mathfrak{M}_2$ implies
$$
b\,\lambda^{-}_k \left|\begin{array}{c}
\!\mu_2\! \\
\!\!-s\!\!
\end{array}\right\rangle
\!\!\left\langle\begin{array}{c}
\!x\! \\
\!y\!
\end{array}\right|=-\sigma_{\!*}\bar{c}\,
\lambda^{+}_k \left|\begin{array}{c}
\!p\! \\
\!q\!
\end{array}\right\rangle
\!\!\left\langle\begin{array}{c}
\!\bar{\mu}_2\! \\
\!\!-s\!\!
\end{array}\right|
\quad \textup{if}\;\; t=1
$$
and
$$
 a\,\lambda^{+}_k\left|\begin{array}{c}
\!\mu_2\! \\
\!s\!
\end{array}\right\rangle
\!\!\left\langle\begin{array}{c}
\!x\! \\
\!y\!
\end{array}\right|=-\sigma_{\!*}\bar{d}\,
\lambda^{-}_k \left|\begin{array}{c}
\!p\! \\
\!q\!
\end{array}\right\rangle
\!\!\left\langle\begin{array}{c}
\!\bar{\mu}_2\! \\
\!\!s\!
\end{array}\right|
\quad \textup{if}\;\; t=-1,
$$
where $\,\lambda^{\pm}_1=\left\langle\begin{array}{c}
\!\bar{\mu}_1\! \\
\!\pm1\!
\end{array}\right|U_1
\left|\begin{array}{c}
\!\mu_1\! \\
\!\pm1\!
\end{array}\right\rangle=2\mu_1^2\,$ and $\,\lambda^{\pm}_2=\left\langle\begin{array}{c}
\!\!\bar{\mu}_1\! \\
\!\!\pm1\!\!
\end{array}\right|U_2
\left|\begin{array}{c}
\!\!\mu_1\! \\
\!\!\pm1\!\!
\end{array}\right\rangle=\pm2\mu_1$. Since $\,\lambda^{+}_1=\lambda^{-}_1\neq0\,$ and $\,\lambda^{+}_2=-\lambda^{-}_2\neq0$, validity of the above
equalities for $\,k=1,2\,$ implies $\,b=c=0\,$ if $\,t=1\,$ and
$\,a=d=0\,$ if $\,t=-1$. So, $\,x_1,z_2,z_3,y_4$ are product
vectors.\smallskip

2) The pair $(z_2,z_3)$ has form 2, i.e.
$$
z_2=\left[\begin{array}{c}
\!p\! \\
\!q\!
\end{array}\right]\otimes\left[\begin{array}{c}
\!\bar{\mu}_2\! \\
\!t\!
\end{array}\right]\!,\; z_3=
\left[\begin{array}{c}
\!x\! \\
\!y\!
\end{array}\right]\otimes\left[\begin{array}{c}
\!\mu_2\! \\
\!\!-t\!\!
\end{array}\right]\!,\;\; t=\pm1,\,|p|+|q|\neq0,\,|x|+|y|\neq0.
$$
By substituting the expressions for $\,x_1,z_2,z_3,y_4\,$ into
(\ref{nbe-3}) and by noting that
\begin{equation*}
\left\langle\begin{array}{c}
\!\bar{\mu}_2\! \\
\!t\!
\end{array}\right|V_k
\left|\begin{array}{c}
\!\mu_2\! \\
\!\!-t\!\!
\end{array}\right\rangle=0,\;\; t=\pm1,\;\; k=1,2
\end{equation*}
we obtain
$$
a \left\langle\begin{array}{c}
\!p\! \\
\!q\!
\end{array}\right|A
\left|\begin{array}{c}
\!\mu_1\! \\
\!1\!
\end{array}\right\rangle\!\!
\left\langle\begin{array}{c}
\!\bar{\mu}_2\! \\
\!t\!
\end{array}\right|V_k
\left|\begin{array}{c}
\!\mu_2\! \\
\!t\!
\end{array}\right\rangle
=\sigma_{\!*}\bar{c} \left\langle\begin{array}{c}
\!\bar{\mu}_1\! \\
\!1\!
\end{array}\right|A
\left|\begin{array}{c}
\!x\! \\
\!y\!
\end{array}\right\rangle\!\!
\left\langle\begin{array}{c}
\!\bar{\mu}_2\! \\
\!\!-t\!\!
\end{array}\right|V_k
\left|\begin{array}{c}
\!\mu_2\! \\
\!\!-t\!
\end{array}\right\rangle
\quad \textup{if}\;\; t=s
$$ and
$$
b \left\langle\begin{array}{c}
\!p\! \\
\!q\!
\end{array}\right|A
\left|\begin{array}{c}
\!\mu_1\! \\
\!\!-1\!\!
\end{array}\right\rangle\!\!
\left\langle\begin{array}{c}
\!\bar{\mu}_2\! \\
\!t\!
\end{array}\right|V_k
\left|\begin{array}{c}
\!\mu_2\! \\
\!t\!
\end{array}\right\rangle
=\sigma_{\!*}\bar{d} \left\langle\begin{array}{c}
\!\bar{\mu}_1\! \\
\!\!-1\!
\end{array}\right|A
\left|\begin{array}{c}
\!x\! \\
\!y\!
\end{array}\right\rangle\!\!
\left\langle\begin{array}{c}
\!\bar{\mu}_2\! \\
\!\!-t\!
\end{array}\right|V_k
\left|\begin{array}{c}
\!\mu_2\! \\
\!\!-t\!\!
\end{array}\right\rangle
\quad \textup{if}\;\; t=-s.
$$
Validity of this equality for all $A\in\mathfrak{M}_2$ implies
$$
a\,\nu^t_k \left|\begin{array}{c}
\!\mu_1\! \\
\!1\!
\end{array}\right\rangle
\!\!\left\langle\begin{array}{c}
\!p\! \\
\!q\!
\end{array}\right|=\sigma_{\!*}
\bar{c}\,\nu^{-t}_k \left|\begin{array}{c}
\!x\! \\
\!y\!
\end{array}\right\rangle
\!\!\left\langle\begin{array}{c}
\!\bar{\mu}_1\! \\
\!1\!
\end{array}\right|
\quad \textup{if}\;\; t=s
$$
and
$$
b\,\nu^{t}_k \left|\begin{array}{c}
\!\mu_1\! \\
\!\!-1\!\!
\end{array}\right\rangle
\!\!\left\langle\begin{array}{c}
\!p\! \\
\!q\!
\end{array}\right|=\sigma_{\!*}
\bar{d}\,\nu^{-t}_k \left|\begin{array}{c}
\!x\! \\
\!y\!
\end{array}\right\rangle
\!\!\left\langle\begin{array}{c}
\!\bar{\mu}_1\! \\
\!\!-1\!\!
\end{array}\right|
\quad \textup{if}\;\; t=-s,
$$
where $\,\nu^{t}_1=\left\langle\begin{array}{c}
\!\bar{\mu}_2\! \\
\!t\!
\end{array}\right|V_1
\left|\begin{array}{c}
\!\mu_2\! \\
\!t\!
\end{array}\right\rangle=2\mu_2^2\,$ and $\,\nu^{t}_2=\left\langle\begin{array}{c}
\!\!\bar{\mu}_2\! \\
\!\!t\!\!
\end{array}\right|V_2
\left|\begin{array}{c}
\!\!\mu_2\! \\
\!\!t\!\!
\end{array}\right\rangle=2t\mu_2$. Since $\,\nu^{t}_1=\nu^{-t}_1\neq0\,$ and
$\,\nu^{t}_2=-\nu^{-t}_2\neq0$, validity of the above equalities for
$\,k=1,2\,$ implies  $\,a=c=0\,$ if $\,t=s\,$ and $\,b=d=0\,$ if
$\,t=-s$. So, $\,x_1,z_2,z_3,y_4\,$ are product vectors.\smallskip

3) The pair $(z_2,z_3)$ has form 3, i.e.
$$ z_2=p\left[\begin{array}{c}
\!\mu_1\! \\
\!1\!
\end{array}\right]\otimes\left[\begin{array}{c}
\!\bar{\mu}_2\! \\
\!t\!
\end{array}\right]+
q\left[\begin{array}{c}
\!\mu_1\! \\
\!\!-1\!\!
\end{array}\right]\otimes\left[\begin{array}{c}
\!\bar{\mu}_2\! \\
\!\!-t\!
\end{array}\right]\!,
\; z_3=x\left[\begin{array}{c}
\!\bar{\mu}_1\! \\
\!\!-1\!\!
\end{array}\right]\otimes\left[\begin{array}{c}
\!\mu_2\! \\
\!t\!
\end{array}\right]+
y\left[\begin{array}{c}
\!\bar{\mu}_1\! \\
\!1\!
\end{array}\right]\otimes\left[\begin{array}{c}
\!\mu_2\! \\
\!\!-t\!\!
\end{array}\right]\!,
$$
where $t=\pm1$. If we substitute the expressions for
$\,x_1,z_2,z_3,y_4\,$ into (\ref{nbe-2}) (by using (\ref{note}))
then the left and the right hand sides of this equality will be
equal respectively to
$$
\bar{x}b \left\langle\begin{array}{c}
\!\bar{\mu}_1\! \\
\!\!-1\!\!
\end{array}\right|U_k
\left|\begin{array}{c}
\!\mu_1\! \\
\!\!-1\!\!
\end{array}\right\rangle\!\!\left\langle\begin{array}{c}
\!\mu_2\! \\
\!t\!
\end{array}\right|A
\left|\begin{array}{c}
\!\mu_2\! \\
\!\!-s\!\!
\end{array}\right\rangle
+\bar{y}a \left\langle\begin{array}{c}
\!\bar{\mu}_1\! \\
\!1\!
\end{array}\right|U_k
\left|\begin{array}{c}
\!\mu_1\! \\
\!1\!
\end{array}\right\rangle\!\!\left\langle\begin{array}{c}
\!\mu_2\! \\
\!\!-t\!\!
\end{array}\right|A
\left|\begin{array}{c}
\!\mu_2\! \\
\!s\!
\end{array}\right\rangle
$$
and to
$$
-\sigma_{\!*}\bar{c}p \left\langle\begin{array}{c}
\!\bar{\mu}_1\! \\
\!1\!
\end{array}\right|U_k
\left|\begin{array}{c}
\!\mu_1\! \\
\!1\!
\end{array}\right\rangle\left\langle\begin{array}{c}
\!\bar{\mu}_2\! \\
\!\!-s\!\!
\end{array}\right|A
\left|\begin{array}{c}
\!\bar{\mu}_2\! \\
\!t\!
\end{array}\right\rangle-\sigma_{\!*}
\bar{d}q \left\langle\begin{array}{c}
\!\bar{\mu}_1\! \\
\!\!-1\!\!
\end{array}\right|U_k
\left|\begin{array}{c}
\!\mu_1\! \\
\!\!-1\!\!
\end{array}\right\rangle\left\langle\begin{array}{c}
\!\bar{\mu}_2\! \\
\!s\!
\end{array}\right|A
\left|\begin{array}{c}
\!\bar{\mu}_2\! \\
\!\!-t\!\!
\end{array}\right\rangle.
$$
So, validity of the equality for all $A\in\mathfrak{M}_2$ implies
$$
\!\left[\bar{y}a \left|\begin{array}{c}
\!\mu_2\!\! \\
\!s\!
\end{array}\right\rangle
\!\!\left\langle\begin{array}{c}
\!\!\mu_2\!\! \\
\!\!-t\!
\end{array}\right|+\sigma_{\!*}
\bar{c}p \left|\begin{array}{c}
\!\bar{\mu}_2\!\! \\
\!t\!
\end{array}\right\rangle
\!\!\left\langle\begin{array}{c}
\!\!\bar{\mu}_2\!\! \\
\!\!-s\!
\end{array}\right|\right]
=\varsigma_k\!\left[\sigma_{\!*}\bar{d}q \left|\begin{array}{c}
\!\bar{\mu}_2\! \\
\!\!-t\!\!
\end{array}\right\rangle
\!\!\left\langle\begin{array}{c}
\!\bar{\mu}_2\!\! \\
\!s\!
\end{array}\right|+
\bar{x}b \left|\begin{array}{c}
\!\mu_2\! \\
\!\!-s\!\!
\end{array}\right\rangle
\!\!\left\langle\begin{array}{c}
\!\mu_2\!\! \\
\!t\!
\end{array}\right|\right]
$$
where $\varsigma_k\doteq-\lambda^{-}_k/\lambda^{+}_k=(-1)^k$. This
equality can be valid for $\,k=1,2\,$ only if the operators in the
squared brackets are equal to zero. Since $\mu_2\neq\pm\bar{\mu}_2$
by the assumption $\theta_2\neq0,\pi$, it follows that
$\,ya=cp=dq=xb=0$. It is easy to see that this implies that $\,x_1,
z_2, z_3, y_4\,$ are product vectors.\medskip

The similar argumentation shows incompatibility of the system
(\ref{nbe-1})-(\ref{nbe-5}) (by reducing to the case of tensor
product vectors) if the pair $(z_2,z_3)$ has form 3 and the pair
$(x_1,y_4)$ has form 1 or 2.\medskip

Assume finally that the pair $(x_1,y_4)$ has form 4 , i.e.
$$
x_1=h\left[\begin{array}{c}
\!\mu_1\! \\
\!s\!
\end{array}\right]\otimes\left[\begin{array}{c}
\!\mu_2\! \\
\!t\!
\end{array}\right],
\; y_4=\left[\begin{array}{c}
\!\bar{\mu}_1\! \\
\!\!-s\!\!
\end{array}\right]\otimes\left[\begin{array}{c}
\!a\! \\
\!b\!
\end{array}\right]+
\left[\begin{array}{c}
\!c\! \\
\!d\!
\end{array}\right]\otimes\left[\begin{array}{c}
\!\bar{\mu}_2\! \\
\!\!-t\!\!
\end{array}\right], \quad s,t=\pm1,
$$
and the pair $(z_2,z_3)$ is arbitrary. We will show that
(\ref{nbe-1})-(\ref{nbe-3}) imply that $y_4$ is a product vector.
So, in fact the pair $(x_1,y_4)$ has form 1 or 2.

Assume  that $\,y_4\,$ is not a product vector and  denote the
vectors $\,[\mu_1, s]^{\top}$ and $\,[\mu_2, t]^{\top}$ by
$|s\rangle$ and $|t\rangle$. In this notations
$\,|x_1\rangle=h|s\otimes t\rangle$.

In case a) it follows from (\ref{nbe-2}) and Lemma \ref{al-3} in the
Appendix that\break $|x_2\rangle=|p\otimes t\rangle\,$ for some
vector $|p\rangle$. Hence the left hand side of (\ref{nbe-1}) has
the form
$$
|h|^2|s\rangle\langle s|\otimes|t\rangle\langle t|+|p\rangle\langle
p|\otimes|t\rangle\langle t|=\left[|h|^2|s\rangle\langle
s|+|p\rangle\langle p|\right]\otimes|t\rangle\langle t|
$$
and (\ref{nbe-1}) implies $\,|y_4\rangle\langle
y_4|\leq\left[|h|^2|s\rangle\langle s|+|p\rangle\langle
p|\right]\otimes|t\rangle\langle t|$. This operator inequality can
be valid only if $\,y_4\,$ is a product vector.

In case b) it follows from (\ref{nbe-3}) and Lemma \ref{al-3} in the
Appendix that $\,|x_3\rangle=|s\otimes q\rangle\,$ for some vector
$|q\rangle$. Hence the left hand side of (\ref{nbe-1}) has the form
$$
|h|^2|s\rangle\langle s|\otimes|t\rangle\langle t|+|s\rangle\langle
s|\otimes|q\rangle\langle q|=|s\rangle\langle
s|\otimes\left[|h|^2|t\rangle\langle t|+|q\rangle\langle q|\right]
$$
and similarly to the case a) we conclude that $\,y_4\,$ is a product
vector.
\medskip

By using the same argumentation exploiting
(\ref{nbe-1})-(\ref{nbe-3}) and Lemma \ref{al-3} one can show that
neither $(x_1,y_4)$ nor $(z_2,z_3)$ can be a pair of form 4 or 5
(not coinciding with form 1 or 2).\medskip

Thus, we have shown that the system (\ref{b-eq-2})-(\ref{b-eq-10})
has no nontrivial solutions. This completes the proof of assertion
$\mathrm{C_2}$.  $\square$

\section*{Appendix}

\subsection{Proofs of Lemmas \ref{sys-l-1} and \ref{sys-l-2}}

\emph{Proof of Lemma \ref{sys-l-1}.} A) Let $\,\langle
z_4|=[a,b,c,d]\,$ and
$$
W=\left[\begin{array}{cccc}
a &  \gamma_2 b & \gamma_1 c &  \gamma_1\gamma_2 d\\
b &  a & \gamma_1 d & \gamma_1 c\\
c &  \gamma_2 d & a &  \gamma_2 b\\
d &  c & b &  a
\end{array}\right],\quad
S=\left[\begin{array}{rrrr}
\mu_1\mu_2 &  \mu_1\mu_2 & \mu_1\mu_2 &  \mu_1\mu_2\\
\mu_1 &  -\mu_1 & \mu_1 & -\mu_1\\
\mu_2 &  \mu_2 & -\mu_2 &  -\mu_2\\
+1 &  -1 & -1 &  +1
\end{array}\right],
$$
where $\mu_k=\sqrt{\gamma_k},\,k=1,2$. By identifying $A\otimes B$
with the matrix $\|a_{ij}B\|$ the equalities $\;\langle
z_4|U_k\otimes V_l|z_1\rangle=0,\;k,l=1,2\;$ can be rewritten as the
system  of linear equations
\begin{equation}\label{as-1}
W|z_1\rangle=0
\end{equation}
and it is easy to see that
$S^{-1}WS=\mathrm{diag}\{p_1,p_2,p_3,p_4\}$, where
\begin{equation}\label{p-def}
\begin{array}{c}
p_1=a+\mu_2 b+\mu_1 c+\mu_1\mu_2 d,\qquad p_2=a-\mu_2 b+\mu_1 c-\mu_1\mu_2 d,\\
p_3=a+\mu_2 b-\mu_1 c-\mu_1\mu_2 d,\qquad p_4=a-\mu_2 b-\mu_1
c+\mu_1\mu_2 d.
\end{array}
\end{equation}

So, system (\ref{as-1}) can be  rewritten as the system $p_ku_k=0$,
$k=\overline{1,4}$, where
$[u_1,u_2,u_3,u_4]^{\top}=S^{-1}|z_1\rangle$. Hence this system has
nontrivial solutions if and only if $\;p_1p_2p_3p_4=0\;$ and
$$
\{p_k=0\}\Leftrightarrow\{W|q_k\rangle=0\},
$$
where $|q_k\rangle$ is the $k$-th column of the matrix $S$.

Thus, by choosing some of $\,p_1,\ldots, p_4\,$ equal to zero we
obtain all pairs $(z_1,z_4)$ such that $\;\langle z_4|U_k\otimes
V_l|z_1\rangle=0,\;k,l=1,2\;$. We have
\begin{enumerate}[a)]
    \item $C_4^2=6\;$ possibilities to take $\;p_k=p_l=0\;$ and $\;p_i\neq0,\; i\neq
    k,l$;
    \item $C_4^1=4\;$ possibilities to take $\;p_k=0\;$ and $\;p_i\neq0, \;i\neq
    k$;
    \item $C_4^3=4\;$ possibilities to take $\;p_k=p_l=p_j=0\;$ and $\;p_i\neq0,\; i\neq
    k,l,j$.
\end{enumerate}
(the case $\;p_1=p_2=p_3=p_4=0\;$ means that $\;a=b=c=d=0\;$, so it
gives only trivial solution). \smallskip

By identifying the vectors $\,x\otimes y\,$ and  $\,[x_1y,
x_2y]^{\top}\,$ it is easy to see that
$$
|q_1\rangle\!=\!\left[\begin{array}{c}
\!\!\mu_1\!\! \\
\!1\!
\end{array}\right]\otimes\left[\begin{array}{c}
\!\!\mu_2\!\! \\
\!1\!
\end{array}\right]\!\!,\,
|q_2\rangle\!=\!\left[\begin{array}{c}
\!\!\mu_1\!\! \\
\!1\!
\end{array}\right]\otimes\left[\begin{array}{c}
\!\!\mu_2\!\! \\
\!\!-1\!\!
\end{array}\right]\!\!,\,
|q_3\rangle\!=\!\left[\begin{array}{c}
\!\!\mu_1\!\! \\
\!\!-1\!\!
\end{array}\right]\otimes\left[\begin{array}{c}
\!\!\mu_2\!\! \\
\!1\!
\end{array}\right]\!\!,\,
|q_4\rangle\!=\!\left[\begin{array}{c}
\!\!\mu_1\!\! \\
\!\!-1\!\!
\end{array}\right]\otimes\left[\begin{array}{c}
\!\!\mu_2\!\! \\
\!\!-1\!\!
\end{array}\right]
$$
and that
$$
p_1=0\quad\Leftrightarrow\quad |z_4\rangle=\left[\begin{array}{c}
c_1\! \\
c_2\!
\end{array}\right]\otimes\left[\begin{array}{c}
\!\bar{\mu}_2\! \\
\!\!-1\!\!
\end{array}\right]+\left[\begin{array}{c}
\!\bar{\mu}_1\! \\
\!\!-1\!\!
\end{array}\right]\otimes\left[\begin{array}{c}
c_3\! \\
c_4\!
\end{array}\right],\quad c_1,\ldots c_4\in\mathbb{C},
$$
$$
p_2=0\quad\Leftrightarrow\quad |z_4\rangle=\left[\begin{array}{c}
c_1\! \\
c_2\!
\end{array}\right]\otimes\left[\begin{array}{c}
\!\bar{\mu}_2\! \\
1
\end{array}\right]+\left[\begin{array}{c}
\!\bar{\mu}_1\! \\
\!\!-1\!\!
\end{array}\right]\otimes\left[\begin{array}{c}
c_3\! \\
c_4\!
\end{array}\right],\quad c_1,\ldots c_4\in\mathbb{C},
$$
$$
p_3=0\quad\Leftrightarrow\quad |z_4\rangle=\left[\begin{array}{c}
c_1\! \\
c_2\!
\end{array}\right]\otimes\left[\begin{array}{c}
\!\bar{\mu}_2\! \\
\!\!-1\!\!
\end{array}\right]+\left[\begin{array}{c}
\!\bar{\mu}_1\! \\
1
\end{array}\right]\otimes\left[\begin{array}{c}
c_3\! \\
c_4\!
\end{array}\right],\quad c_1,\ldots c_4\in\mathbb{C},
$$
$$
p_4=0\quad\Leftrightarrow\quad |z_4\rangle=\left[\begin{array}{c}
c_1\! \\
c_2\!
\end{array}\right]\otimes\left[\begin{array}{c}
\!\bar{\mu}_2\! \\
1
\end{array}\right]+\left[\begin{array}{c}
\!\bar{\mu}_1\! \\
1
\end{array}\right]\otimes\left[\begin{array}{c}
c_3\! \\
c_4\!
\end{array}\right],\quad c_1,\ldots c_4\in\mathbb{C},
$$
Hence the above six possibilities in a) correspond to cases 1)-3) in
Lemma \ref{sys-l-1}A (for example, the choice $\,p_1=p_2=0,\,$
$\,p_3,p_4\neq0\,$ corresponds to case 1) with $s=1$), while the
four possibilities in b) and in c) correspond respectively to cases
4) and 5). \medskip

B) Denote by $W_x$ and $W_y$ the above matrix $W$ determined
respectively via $z_4=x_4$ and $z_4=y_4$. Then the equalities in
(\ref{b-eq-7}) and (\ref{b-eq-8}) can be rewritten as the system
\begin{equation}\label{s-tmp}
    W_x|y_1\rangle=W_y|x_1\rangle=0,\qquad
    W_x|x_1\rangle=W_y|y_1\rangle=|c\rangle, \qquad
    |c\rangle\in\mathbb{C}^4.
\end{equation}
Since $S^{-1}W_xS=\mathrm{diag}\{p^x_1,p^x_2,p^x_3,p^x_4\}$ and
$S^{-1}W_yS=\mathrm{diag}\{p^y_1,p^y_2,p^y_3,p^y_4\}$, where
$p^x_1,p^x_2,p^x_3,p^x_4$ and $p^y_1,p^y_2,p^y_3,p^y_4$ are defined
in (\ref{p-def}) with $z_4=x_4$ and $z_4=y_4$ correspondingly,
system (\ref{s-tmp}) is equivalent to the following one
\begin{equation}\label{s-tmp-2}
p^x_k v_k=p^y_k u_k=0,\qquad p^x_k u_k=p^y_k v_k=\tilde{c}_k,\qquad
k=\overline{1,4},
\end{equation}
where $[u_1,u_2,u_3,u_4]^{\top}\!=S^{-1}|x_1\rangle$,
$[v_1,v_2,v_3,v_4]^{\top}\!=S^{-1}|y_1\rangle$ and
$[\tilde{c}_1,\tilde{c}_2,\tilde{c}_3,\tilde{c}_4]^{\top}\!=S^{-1}|c\rangle$.
It follows that $\tilde{c}_k=0$ for all $k$. Indeed, if $p^y_k\neq0$
for some $k$ then the first equality in (\ref{s-tmp-2}) implies
$u_k=0$ and the second equality in (\ref{s-tmp-2}) shows that
$\tilde{c}_k=0$. Hence $|c\rangle=S|\tilde{c}\rangle=0$.
$\square$\medskip

Lemma \ref{sys-l-2} follows from Lemma \ref{sys-l-1} with $\gamma_2$
replaced by $\bar{\gamma}_2$.

\subsection{Auxiliary lemmas}

\begin{lemma}\label{al-1}
\emph{If $\,|a\rangle\langle x|+|b\rangle\langle y|+|c\rangle\langle
z|=0\,$ then either $\;a\!\parallel\! b\!\parallel\! c\,$ or
$\;x\!\parallel\! y\!\parallel\! z$.}
\end{lemma}

\begin{proof} We may assume that all the vectors are nonzero (since otherwise the assertion is trivial).

Let $p\perp x$. Then $\langle y|p\rangle|b\rangle+\langle
z|p\rangle|c\rangle=0$ and hence either $b\parallel c$ or $\langle
y|p\rangle=\langle z|p\rangle=0$.

If $b\parallel c$ then we have $|a\rangle\langle
x|=-|b\rangle\langle y+\lambda z|,\, \lambda\in\mathbb{C},$ and
hence $a\parallel b\parallel c$.

If $\langle y|p\rangle=\langle z|p\rangle=0$ then $x\parallel
y\parallel z$, since the vector $p$ is arbitrary.  \end{proof}

\begin{lemma}\label{al-2}
\emph{The equality
\begin{equation}\label{al-2-e}
X_1\otimes Y_1+X_2\otimes Y_2=X_3\otimes Y_3+X_4\otimes Y_4,
\end{equation}
where $X_i=|x_i\rangle\langle x_i|,\;Y_i=|y_i\rangle\langle y_i|,\;
i=\overline{1,4}$, can be valid only in the following cases:}
\begin{enumerate}[1)]
    \item \emph{$x_i\parallel x_j\,$ for all $\;i,j\,$ and
    $\;Y_1\|x_1\|^2+Y_2\|x_2\|^2=Y_3\|x_3\|^2+Y_4\|x_4\|^2$};
    \item \emph{$y_i\parallel y_j\,$ for all $\;i,j\,$ and
    $\;X_1\|y_1\|^2+X_2\|y_2\|^2=X_3\|y_3\|^2+X_4\|y_4\|^2$};
    \item \emph{$X_1\otimes Y_1=X_4\otimes Y_4\,$ and $\;X_2\otimes Y_2=X_3\otimes
    Y_3$};
    \item \emph{$X_1\otimes Y_1=X_3\otimes Y_3\,$ and $\;X_2\otimes Y_2=X_4\otimes
    Y_4$}.
    \end{enumerate}
\end{lemma}

\begin{proof} We may assume that all the vectors $\,x_i,y_i\,$ are nonzero (since otherwise the assertion is trivial).

Let $p\perp x_1$. By multiplying the both sides of (\ref{al-2-e}) by
$|p\rangle\langle p|\otimes I$ we obtain
\begin{equation}\label{al-2-e+}
|\langle x_2|p \rangle|^2Y_2=|\langle x_3|p \rangle|^2Y_3+|\langle
x_4|p \rangle|^2Y_4.
\end{equation}
If $\,x_2\parallel x_1\,$ then $\,\langle x_3|p \rangle=\langle
x_4|p \rangle=0\,$ and hence $\,x_1\parallel x_2\parallel
x_3\parallel x_4$, since the vector $p$ is arbitrary. So, case 1)
holds.

If $\,x_2\nparallel x_1\,$ then one can choose $p$ such that
$\,\langle x_2|p \rangle\neq0$. So, (\ref{al-2-e+}) implies that
either $\,x_3\nparallel x_1\,$ or $\,x_4\nparallel x_1\,$. Thus, we
have the following possibilities:\smallskip

a) If $\,x_i\nparallel x_1\,$ for $\,i=2,3,4\,$ then one can choose
$p$ such that $\,\langle x_i|p \rangle\neq0,$ $\,i=2,3,4,$ and
(\ref{al-2-e+}) implies $\,y_2\parallel y_3\parallel y_4$. Hence
(\ref{al-2-e}) leads to $\,X_1\otimes Y_1=[\ldots]\otimes Y_2$,
which gives $\,y_1\parallel y_2$. So, we have $\,y_1\parallel
y_2\parallel y_3\parallel y_4$, i.e. case 2).\smallskip

b) If $\,x_i\nparallel x_1\,$ for $\,i=2,3\,$ but $\,x_4\parallel
x_1\,$ then one can choose $p$ such that $\,\langle x_i|p
\rangle\neq0$, $\,i=2,3$ and (\ref{al-2-e+}) implies $\,y_2\parallel
y_3$. So, we have $\,x_4=\alpha x_1\,$ and $\,y_3=\beta y_2$,
$\alpha,\beta\in\mathbb{C}$. It follows from (\ref{al-2-e}) that
$$
X_1\otimes [Y_1-|\alpha|^2Y_4]=[X_3|\beta|^2-X_2]\otimes Y_2
$$
and hence $Y_1-|\alpha|^2Y_4=\lambda Y_2$, $\lambda\in\mathbb{C}$.
If $\lambda\neq0$ then Lemma \ref{al-1} implies $y_1\parallel
y_2\parallel y_3\parallel y_4$, i.e. case 2) holds. If $\lambda=0$
then $y_1\parallel y_4$ and $x_2\parallel x_3$. Thus we have
$$
X_4\otimes Y_4=\gamma  X_1\otimes Y_1,\;\;  X_3\otimes Y_3= \delta
X_2\otimes Y_2,\quad \gamma,\delta\in\mathbb{C}
$$
and (\ref{al-2-e}) implies $(1-\gamma)X_1\otimes
Y_1=(\delta-1)X_2\otimes Y_2$. Since $x_1\nparallel x_2$, we have
$\gamma=\delta=1$, i.e. case 3) holds.\smallskip

c) If $\,x_i\nparallel x_1\,$ for $\,i=2,4\,$ but $\,x_3\parallel
x_1\,$ then the similar arguments (with the permutation
$3\leftrightarrow 4$) shows that case 4) holds.

\end{proof}

\begin{lemma}\label{al-3}
\emph{Let $\,U_1=\mathrm{diag}\{1,\gamma\}$ and $\,x,y$ be nonzero
vectors in $\,\mathbb{C}^2$. If $\,\langle a|\, U_1\otimes A
|\,x\otimes y\rangle=\langle c|\, U_1\otimes A |d \rangle\,$ for all
$\,A\in\mathfrak{M}_2$ then either
$\,|d\rangle=|z\rangle\otimes|y\rangle$ or
$\,|c\rangle=|p\rangle\otimes|q\rangle$ for some vectors $\,p,q,z$
in $\,\mathbb{C}^2$.}
\end{lemma}

\begin{proof}
By using the isomorphism $\mathbb{C}^2\otimes\mathbb{C}^2\ni
u\otimes v \leftrightarrow [u_1v,
u_2v]^{\top}\in\mathbb{C}^2\oplus\mathbb{C}^2$ the condition of the
lemma can be rewritten as follows
$$
\left\langle\begin{array}{c}
\!a_1\! \\
\!a_2\!
\end{array}\right|\left.\begin{array}{cc}
\!A\! & \!0\! \\
\!0\! & \!\gamma A\!
\end{array}\right|
\left.\begin{array}{c}
\!x_1 y\! \\
\!x_2 y\!
\end{array}\right\rangle=
\left\langle\begin{array}{c}
\!c_1\! \\
\!c_2\!
\end{array}\right|\left.\begin{array}{cc}
\!A\! & \!0\! \\
\!0\! & \!\gamma A\!
\end{array}\right|
\left.\begin{array}{c}
\!d_1\! \\
\!d_2\!
\end{array}\right\rangle,\quad\forall A\in \mathfrak{M}_2,
$$
where $a_1, a_2$ are components of the vector $a$, etc. So, we have
$$
x_1 \langle a_1|A|y\rangle+x_2 \gamma\langle a_2|A|y\rangle=\langle
c_1|A|d_1\rangle+\gamma\langle c_2|A|d_2\rangle\quad\forall A\in
\mathfrak{M}_2,
$$
which is equivalent to the equality $\,|y\rangle\langle\bar{x}_1
a_1+\bar{x}_2\bar{\gamma}a_2|=|d_1\rangle\langle
c_1|+\gamma|d_2\rangle\langle c_2|$. By Lemma \ref{al-1} it follows
that either $d_1\parallel d_2 \parallel y$, which means that
$\,|d\rangle=|z\rangle\otimes|y\rangle$, or $c_1\parallel c_2$,
which means that $\,|c\rangle=|p\rangle\otimes|q\rangle$.
\end{proof}

\bigskip

I am grateful to A.S.Holevo and to the participants of his seminar
"Quantum probability, statistic, information" (the Steklov
Mathematical Institute) for useful discussion.

\end{document}